\documentclass[12pt]{article}%
\usepackage{amsfonts}
\usepackage{amssymb}
\usepackage{graphicx}
\usepackage{setspace}
\usepackage[round]{natbib}
\usepackage{amsmath}%
\usepackage{amsthm}%
\usepackage{relsize}
\usepackage{subcaption}
\usepackage{comment}
\usepackage{palatino}
\usepackage{paralist}
\usepackage{url}

\usepackage[top=6pc,bottom=6pc,left=6pc,right=6pc]{geometry}

\newtheorem{theorem}{Theorem}

\newtheorem{corollary}[theorem]{Corollary}

\newtheorem{definition}[theorem]{Definition}

\newtheorem{lemma}[theorem]{Lemma}

\newtheorem{assumption}{Assumption}

\newcommand{\by}{{\bf y}}

\title{ A deconvolution path for mixtures}
\author{Oscar-Hernan Madrid-Padilla\footnote{Ph.D student; Department of Statistics and Data Sciences, University of Texas at Austin; oscar.madrid@utexas.edu} \\
	Nicholas G.~Polson\footnote{Professor; University of Chicago Booth School of Business; ngp@chicagobooth.edu} \\
	James Scott\footnote{Associate Professor; Department of Information, Risk, and Operations Management and Department of Statistics and Data Sciences, University of Texas at Austin; james.scott@mccombs.utexas.edu}}
\date{\today}

\begin{spacing}{1.5}

\begin{document}

\maketitle

\begin{abstract}
	We propose a class of estimators for deconvolution in mixture models based on a simple two-step ``bin-and-smooth'' procedure applied to histogram counts.  The method is both statistically and computationally efficient: by exploiting recent advances in convex optimization, we are able to provide a full deconvolution path that shows the estimate for the mixing distribution across a range of plausible degrees of smoothness, at far less cost than a full Bayesian analysis.  This enables practitioners to conduct a sensitivity analysis with minimal effort. This is especially important for applied data analysis, given the ill-posed nature of the deconvolution problem.  Our results establish the favorable theoretical properties of our estimator and show that it offers state-of-the-art performance when compared to benchmark methods across a range of scenarios.
	
	\bigskip
	
\end{abstract}

\newpage

\section{Deconvolution in mixture models}

\subsection{Introduction}
Suppose that we observe $\by = \left(y_1,\ldots, y_n\right)$ from the model
\begin{equation}
	\label{eqn:normal_means_model}
	y_i \mid \mu_i \sim \phi(y_i \mid \mu_i) \, , \quad  \mu_i \stackrel{i.i.d.}{\sim} f_0 \, ,
\end{equation}
where $\phi(\cdot \mid \mu)$ is a known distribution with location parameter $\mu$, and $f_0$ is an unknown mixing distribution.  Marginally, we have specified a mixture model for $y_i$:
\begin{equation}
	\label{eqn:y_marginal}
	m(y_i) = \int_\mathbb{R} \phi(y_i - \mu_i) \ f_0(\mu_i) \ d \mu_i = (\phi * f_0)(y_i) \, .
\end{equation}
The problem of estimating the mixing distribution $f_0$ is commonly referred to as \textit{deconvolution}: we observe draws from the convolution $m = \phi * f_0$, rather than from $f_0$ directly, and we wish to invert this blur operation to recover the distribution of the latent location parameters.  Models of this form have been used in a wide variety of applications and have attracted significant attention in the literature \citep[e.g.][]{kiefer:wolfowitz:1956,ferguson1973,fan1991optimal,newton2002nonparametric,ghosal2001entropies}.  Yet the estimation of $f_0$ continues to pose both theoretical and practical challenges, making it an active area of statistical research \citep[e.g.][]{delaigle2014parametrically,efron2016empirical}.

In this paper, we propose a nonparametric method for deconvolution that is both statistically and computationally efficient.  Our method can be motivated in terms of an underlying Bayesian model incorporating a prior into model (\ref{eqn:normal_means_model}), but it does not involve full Bayes analysis.  Rather, we use a two-step ``bin and smooth'' procedure.   In the ``bin'' step, we form a histogram of the sample, yielding the number of observations $x_j$ that fall into the $j$th histogram bin.  In the ``smooth'' step, we use the counts $x_j$ to compute a maximum \textit{a posteriori} (MAP) estimate of $f_0$ under a prior that encourages smoothness.

%$P(\theta) = \sum_{ (i,j) \in E } \rho_{ij}(\theta_i - \theta_j)$, where$ \rho_{ij}$ is a metric.  Note that each $\rho_{ij}(x) = |x|$ for TV, and $\rho_{ij}(x) = 1\{ x \not= 0\}$ for Potts minimization.  Is $P(\theta)$ over the DFS induced chain graph still bounded by $2 P(\theta)$ over the full graph?

We show that this nonparametric empirical-Bayes procedure yields excellent performance for deconvolution, at reduced computational cost compared to full nonparametric Bayesian methods.  Our main theorems establish conditions under which the method yields a consistent estimate of the mixing distribution $f_0$, and provide a concentration bound for recovery of the marginal distribution $m$.  We also provide simulation evidence that the method offers practical improvements over existing state-of-the-art methods.

\section{Connections with previous work}
%Previous approaches to the deconvolution problem (Model \ref{eqn:normal_means_model}) have focused on two inferential goals: (1) estimating the mixing distribution $f_0$, and (2) estimating the means $\mu_i$.  We organize our review of the literature around these two different (albeit related) goals.

%\subsection{Full deconvolution: estimating $f_0$ directly}
%\label{subsec:direct_deconvolution}

We first recall the work by \cite{kiefer:wolfowitz:1956}, who consider estimating $f_0$ using the nonparametric maximum likelihood estimation. The Kiefer--Wolfowitz estimator (KW) has some appealing features: it is completely nonparametric and invariant to translations of the data, it requires no tuning parameters, and it is consistent under fairly general conditions.  Balanced against these desirable features there is, perhaps, a  disadvantage when estimating a smooth  true mixing density: the  KW estimator is a discrete distribution involving as many as $n+1$ point masses.% Hence, such estimator might not be an ideal estimator when the true mixing density is smooth.  %\cite{koenker:mizera:2007} refer to this phenomenon as a ``Dirac catastrophe,'' for the reason that  in most settings $f_0$ will be relatively smooth, and the discreteness of the KW estimator is unappealing. A related consistent estimator  was studied in \cite{geman:hwang:1982}.

%A related estimator for Gaussian observations was studied in \cite{geman:hwang:1982}, which restricts the maximization set in order to obtain consistent estimators.

Our approach  to deconvolution falls within a penalized likelihood framework. It has important connections with classical ideas of penalized likelihood density estimation from \cite{good:gaskins:1971,silverman1982estimation}. And at least in  spirit, our work resembles recent deconvolution methods by \cite{lee2013deconvolution,wager2013geometric,efron2016empirical}.  A  detailed discussion or comparison between these methods and ours  will be made in Section \ref{sec:penalized_likeilhood_approach}.

Deconvolution has also been studied using Bayesian methods. In the context of repeated measurements, or multivariate deconvolution, we highlight recent work by  \cite{sarkar2014bayesian2,sarkar2014bayesian,staudenmayer2008density}. Moreover, for the one dimensional density estimation problem, a flexible choice is the Dirichlet Process (DP) studied in \cite{ferguson1973} and \citet{escobarwest1995}. For a Dirichlet prior in deconvolution problems, concentration rates were recently studied in \cite{donnet2014posterior}. Related models were considered by \cite{domuller2005} and \cite{muralidharan2010empirical} for finite mixture of normals.
%where the marginal density is a weighted sum of an atomic measure and a finite mixture of normals.

The DP provides a very general framework for estimating the mixing density $f_0$.  However, as \cite{martin:tokdar:2012} argue, fitting a Dirichlet process mixture does not scale well with the number of observations $n$. For micro--array  studies, $n$ ranges from thousands to tens of thousands, whereas for more recent studies of fMRI data or single-nucleotide polymorphisms, $n$ can reach several hundreds of thousands \citep[e.g.][]{tansey:etal:2014}. For such large data sets, fitting a DP mixture model can be very time-consuming.

To overcome this difficulty, \cite{newton2002nonparametric}, \cite{tokdar2009consistency}, \cite{martin2011semiparametric}, and \cite{martin:tokdar:2012} studied a predictive recursive (PR) algorithm. The resulting estimator scales well with large data sets while remaining reasonably accurate, thereby solving one the main challenges faced by the fully Bayesian approach.  %Moreover, as described by \cite{martin2011semiparametric} and \cite{martin:tokdar:2012}, this approximate Bayes procedure can be thought as an approximation to DP-model-based inference, and it can easily be adapted to handle multiple testing %where the true mixing density is a convex combination of an atomic measure and an arbitrary density.  The consistency of the PR  algorithm was established by \cite{tokdar2009consistency}.

Finally, we note the work by \cite{carroll1988optimal,stefanski1990deconvolving,zhang1990fourier,fan1991optimal,fan2002wavelet,hall2007ridge,carroll2012deconvolution,delaigle2014parametrically}, among others,  who considered kernel estimators.  
Their idea is motivated by (\ref{eqn:normal_means_model}) after taking the Fourier transform of the corresponding convolution of densities,  then solving for the unknown mixing density using kernel approximations for the Fourier transform of the true marginal density. The resulting kernel estimator enjoys attractive theoretical properties: for each $\mu_0 \in \mathcal{R}$, the estimator  has optimal rates of convergence towards $f_0(\mu_0)$ for squared-error loss when the function $f_0$ belongs to a smooth class of functions \citep{fan1991optimal}. % Unfortunately, the minimax rates are in the order of $(\log n)^{-\alpha}$ for a constant $\alpha > 0$ depending on the smoothness of the functions considered. \cite{carroll1988optimal} comments that the minimax nature of of these convergence rates  may be  pessimistic and perhaps better convergence rates can be found  if we limit consideration to smaller classes of densities, such as those confined to a known interval. 

\section{A deconvolution path}
\label{our_approach}
\subsection{Overview of approach}
We now described our proposed approach in detail.  We study deconvolution estimators related to the variational problem
\begin{equation}
\label{eqn:basic_variational_problem}
 \underset{f}{\text{minimize}}\,\,\,-  \sum_{i=1}^{n} \log (\phi*f)(y_i) \,\,\,\,\,\,\,\, \text{subject to}\,\,\,\,\,\int_{\mathbb{R}} f(\mu)\,d\mu = 1, \,\,\,\,\, J(f) \leq t,
\end{equation}
where $J(f)$ is a known penalty functional. The choices of  $J$  we consider include $\ell_1$ or $\ell_2$ penalties on the derivatives in the log-space to encourage smoothness:
\begin{eqnarray}
\| \log f ^{(k)} \|_s^q   &=& \int_{\mathbb{R}}  \vert \log f ^{(k)}(\mu) \vert^s \ d \mu,
%\| \log f ^{(k)} \|_1 &=& \int_{\mathcal{R}}  \vert \log f^{(k)}(\mu) \vert \ d \mu \, ,
\end{eqnarray}
with $s=q = 1$ or $s= q = 2$ and where $\log f ^{(k)}$ is the derivative of order $k$ of the log prior.  The penalty involving the first derivative is an especially interpretable one, as $ d \log f(\mu)/d\mu = f'(\mu)/f(\mu)$ is the score function of the mixing density.

Note that an alternative interpretation of our approach is as a MAP estimator. To see this  we consider the (possibly improper) prior on the mixing density
$$
p(f)  \propto \exp \left ( - J(f) \right ) \mathbb{I} \left ( f \in \mathcal{A} \right ),
$$
where  $\mathcal{A} $ is an appropriate class of density functions. The posterior distribution is $ p( f \mid y ) \propto p( y \mid f) p( f) $, and our MAP estimator therefore solves, for an appropriate $\tau >0$,
\begin{equation}
\label{MAP_problem}
\text{argmin}_{ f \in \mathcal{A} } - \log p( y \mid f ) \,+\, \frac{\tau}{2} J (f) \, .
\end{equation}
 This belongs to  a general class of MAP estimators that have been studied in  \cite{good:gaskins:1971} and \cite{silverman1982estimation} for the classical problem  of density estimation. For deconvolution problems we note the recent work by \cite{wager2013geometric} which penalizes the marginal density rather than the derivatives of the mixing density as we propose. Moreover such penalization is not motivated to encourage smoothness. Rather,  it is  an $\ell_2$ projection on to the space of acceptable marginal densities.  An alternative penalized likelihood method was studied in \cite{lee2013deconvolution} in the different context  where the marginal density  has atoms. There the authors use the roughness penalty $J(f) = \int \vert f^{\prime}(\mu)\vert^2d\mu$ which differs from our approach that penalizes the log-mixing density and hence ensures that the solutions will be positive. Also, we allow different degrees of smoothness depending on the choice of $k$. Moreover, while  \cite{lee2013deconvolution} only considered sample sizes in the order of hundreds, we show in the next sections that our estimator can scale to much larger data sets while still enjoying attractive statistical properties.
%\[
%\underset{f}{\text{minimize}} -\frac{2}{n}\sum_{i=1}^{n} f(y_i)  +  \|f\|_2^2 \,\,\, \text{subejct to}\,\,\,   \,\,f\,=\,\phi*g,\,\,g\geq 0, \int g \leq 1.
%\]

More recently, \cite{efron2016empirical} proposed a penalized approach to deconvolution that is also based on regularization but differs from ours. Such approach proceeds by assuming that the mixing density is discrete with support $\{\theta_1,\ldots, \theta_{N}\}$. Then \cite{efron2016empirical} specified a parametric model on the mixing density of the form
\[
  f(\theta_j)\,\,=\,\,  \exp\left(Q_{j,\cdot}^T\alpha - c(\alpha)\right),\,\,\,\,j \,\,=\,\, 1,\ldots, N,
\]
where $Q_{j,\cdot}$ is the $j-th$ row of the matrix $Q \in  \mathbb{R}^{N \times p}$, for some $p>0$, which is used to encourage structure on the mixing density. Moreover, $c(\alpha)$ is a normalizing constant satisfying 
\[
    c(\alpha)\,\,=\,\,  \log \left(\sum_{j = 1}^N\, \exp\left(  Q_{j,\cdot}^T\,\alpha\right)  \right).  
\]
Then summing over all the $\{\theta_j\}_{j=1}^N$  and considering the contribution of the different samples, \cite{efron2016empirical} arrives to the minus log-likelihood 
\[
     l(\alpha)\,\, = \,\, -\sum_{i=1}^{n}\,\log\left(  \sum_{j=1}^{N}\, \text{N}(y_i\, |\, \theta_j,1 )\,\exp\left(Q_{j,\cdot}^T\alpha - c(\alpha)\right)   \right).
\]
The estimator from \cite{efron2016empirical}  results from solving 
\begin{equation}
\label{efron_decon}
\underset{\alpha}{\text{minimize}} \,\, \,\,l(\alpha) \,+\, c_0\,\left( \sum_{h=1}^{p} \alpha_h^2 \right)^{1/2}, 
\end{equation}
where $c_0$ is either $1$ or $2$.

We note that the estimator (\ref{efron_decon})  is  possibly limited by the following aspects. First, the choice of the matrix $Q$, which \cite{efron2016empirical}  recommends to be  a spline basis representation, might produce estimates that suffer from local adaptivity problems. Moreover, there is no theoretical support of choosing $c_0 \in \{1,2\}$. In our experiments section we will present experimental comparisons between the estimator (\ref{efron_decon})  and our approach.

Finally, we emphasize that our approach is not, in any sense, related to the ridge parameter deconvolution estimator from 
\cite{hall2007ridge}. Such estimator  does not penalizes that log-likelihood as  we propose, but rather is designed to avoid the need to choose a kernel function, in the original kernel estimator from  \cite{fan1991optimal},  by ridging the integral in its definition with a positive function.

\subsection{Binned counts problem}

Throughout this section we assume that $\phi $ corresponds to the pdf of the standard normal distribution, although the arguments can easily  be generalized to other distributions. 

To make estimation efficient in scenarios with thousands or even millions of observations, we actually fit a MAP estimator based on binning the data.  First, we use the samples to form a histogram $\{I_j , x_j \}_{j=1}^D$ with $D$ bins, where $I_j$ is the $j-$th interval in the histogram and $x_j = \# \{y_i \in I_j\}$ is the associated count.  For ease of exposition, we assume that the intervals take the form $I_j = \xi_j \pm \Delta/2$, i.e.~have midpoints $\xi_j$ and width $\Delta$, although this is not essential to our analysis.  To arrive to a discrete estimator, instead of Problem (\ref{eqn:basic_variational_problem}), we consider an approximation, a reparametrization $g = \log f$, and put the penalty in the objective function with a regulariation parameter $\tau >0$,
\begin{equation}
\label{weighted_likelihood}
\begin{array}{ll}
\underset{ g \in \mathbb{R}^D }{ \text{minimize} } &    -\frac{1}{n} \sum_{j=1}^D x_j \log\left(\, \phi*e^g(\xi_j)   \right) + \frac{\tau}{2}\,J(e^g)  \,\,\,\,\,\,\,\text{subject to} \,\,\,\, \int e^{g(\mu)} d\mu   = 1. \\
%\text{subject to} & \int e^{g(\mu)} d\mu   = 1. \\
% &      ,
\end{array}
\end{equation}
We then approximate (\ref{weighted_likelihood}) by solving
\begin{equation}
\label{weighted_likelihood2}
\begin{array}{ll}
\underset{ g \in \mathbb{R}^{D} }{ \text{minimize} } &    -\frac{1}{n} \sum_{j=1}^D x_j \log\left(\, \sum_{i=1}^D \Delta\,\phi(\xi_j - \xi_i)e^{g_i}    \right) +  \frac{\tau}{2}\,\| \Delta^{(k+1)}g \|_q^s \\
\text{subject to} & \sum_{i=1}^{D_n} \Delta \,e^{g_i}   = 1, \\
% &      ,
\end{array}
\end{equation}
where $s=q=1$ or $s=q=2$, and $\Delta^{(k+1)}$ is the $k$-th order discrete difference operator. Concretely, when $k=0$, $\Delta^{(1)}$ is the $(D-1) \times D$ matrix encoding the first differences of adjacent values:
\begin{equation}
\label{eqn:fusedlassoD}
\Delta^{(1)} =\left(\begin{array}{rrrrrr}
1 & -1 & 0 & 0 & \mathbf{\cdots} & 0\\
0 & 1 & -1 & 0 & \cdots & 0\\
\vdots &  &  &  & \ddots & \vdots \\
0 & \cdots &  & 0 & 1 & -1
\end{array}\right).
\end{equation}
For $k \geq 1$, $\Delta^{(k+1)}$ is defined recursively as $\Delta^{(k+1)} = \Delta^{(1)} \Delta^{(k)}$, where $\Delta^{(1)}$ from (\ref{eqn:fusedlassoD}) is of the appropriate dimension.  Thus when $k=0$, we penalize the total variation of the vector $\theta$ \citep[c.f.][]{rudin:osher:faterni:1992,tibs:fusedlasso:2005} and should expect estimates that are shrunk towards  piecewise-constant functions.  When $k \geq 1$, the estimator penalizes higher-order versions of total variation, similar to the polynomial trend-filtering estimators studied by \citet{tibs:2014a}.

Interestingly, following the proof of Theorem 1 in \cite{padilla2015nonparametric} we find that (\ref{weighted_likelihood2}) is equivalent to
\begin{equation}
\label{problem_for_discrete_estimator}
\begin{aligned}
& \underset{\theta \in \mathcal{R}^D}{\text{minimize}}
& &
l(\theta) + \frac{\tau}{2} \Vert \Delta^{(k+1)} \theta \Vert_q^s \, ,
\end{aligned}
\end{equation}
where
\[
l(\theta) = \sum_{j=1}^D \left\{   \lambda_j(\theta) - x_j \log \lambda_j(\theta)
\right\} \, ,
\]
with $	\lambda_j =  \sum_{j=1}^D G_{ij} e^{\theta_i} $, $G_{ij} = \Delta\,\phi(\xi_j - \xi_i)$, and $\hat{\theta}$ solves (\ref{problem_for_discrete_estimator})  if only if $ \hat{\theta } - \log(n\,\Delta)\boldsymbol{1}$ solves (\ref{weighted_likelihood2}). Hence, in practice we solve the unconstrained optimization Problem (\ref{problem_for_discrete_estimator}).

\subsection{Solution algorithms}

We now discuss the implementation details for solving (\ref{problem_for_discrete_estimator}) in the case $s = q = 1$. To solve this problem, motivated by the work on trend filtering for regression by \cite{ramdas:tibs:2014}, we rewrite the problem as
\begin{equation}
\label{eqn:constrained_admm_version}
\begin{array}{ll}
\underset{\theta }{ \text{minimize}}   & l(\theta) + \frac{\tau}{2} \Vert \Delta^{(1)} \alpha \Vert_1 \,\,\,\,\,\,\text{subject to } \,\,\,\,\,   \alpha = \Delta^{(k)} \theta. 
%\text{subject to } &   \alpha = \Delta^{(k)} \theta.
\end{array}
\end{equation}

Next we proceed via the alternating-direction method of multipliers (ADMM), as in  \citet{ramdas:tibs:2014}.  \citep[See][for an overview of ADMM.]{boyd2011distributed}  By exploiting standard results we arrive at the scaled augmented Lagrangian corresponding to the constrained problem (\ref{eqn:constrained_admm_version}):
\[
L_{\rho }(\theta,\alpha, u)=   l(\theta) + \frac{\tau}{2} \Vert \Delta^{(1)} \alpha \Vert_1  + \rho u^T \left(\alpha - \Delta^{(k) } \theta\right) + \frac{\rho}{2} \|   \alpha +u  - \Delta^{(k)}\theta \|_{2}^{2} \, .
\]
This leads to the following ADMM  updates at each iteration $j$:
\begin{equation}
\label{ADMM}
\begin{array}{lll}
\theta^{j+1}  &  \leftarrow &  \underset{\theta}{ \text{argmin} }  \left(  l(\theta) +  \frac{\rho}{2}\left\| \alpha^{j} + u^{j} - \Delta^{(k)} \theta\right\|_2^{2} \right),\\
\alpha^{j+1} &  \leftarrow & \underset{\alpha}{ \text{argmin} }  \left( \frac{1}{2}\| \alpha - \Delta^{(k)}\theta^{j+1}  + u^{j}\|_2^{2} + \frac{\tau}{2\rho}\| \Delta^{(1)} \alpha\|_1  \right), \\
u^{j+1} &  \leftarrow &   u^{j} +  \alpha^{j+1}  - \Delta^{(k)}\theta^{j+1}.
\end{array}
\end{equation}

Note that in (\ref{ADMM}) the update for $\theta$ involves solving a sub-problem whose solution is not analytically available. To deal with this, we use the well known Broyden--Fletcher--Goldfarb--Shanno (BFGS) algorithm, which is very efficient because the gradient of the $\theta$ sub-problem objective is available in closed form.   The update for $\alpha$ can be computed in linear time by appealing to the  dynamic programming algorithm  from \cite{johnson2013dynamic}.
%for fused lasso problems

In the case $p=q =2$, both components of the objective function in (\ref{problem_for_discrete_estimator})  have closed-form gradients; see the appendix. Thus we can solve the problem using any algorithm that can use function and gradient calls.  In particular, we use BFGS.
%$$
%[\nabla l(\theta)]_j = \sum_{i=1}^D G_{ij} e^{\theta_j} \left( \frac{x_i}{\lambda_i(\theta)} - 1 \right) \, ,
%$$
%and
%$$
%\nabla \Vert \Delta^{(k+1)} \theta \Vert_2^2 = 2 \left(\Delta^{(k+1)} \right)^T \Delta^{(k+1)} \theta \, .
%$$

\subsection{Solution path and model selection}

%subsection{Solution path and model selection}

One of the major advantages of our approach is that it easily yields an entire deconvolution path, comprising a family of estimates $\hat{f}(\tau)$ over a grid of smoothness parameters.  Although in principle any deconvolution algorithm can yield such a path by solving the problem for many smoothing parameters, our path is generated in a highly efficient manner, using warm starts.  We initially solve (\ref{eqn:constrained_admm_version}) for a large value of $\tau$, for which the result estimate is nearly constant.  We then use this solution to initialize the ADMM at a slightly smaller value of $\tau$, which dramatically reduces the computation time compared to an arbitrarily chosen starting point.  We proceed iteratively until solutions have been found across a decreasing grid of $\tau$ values (which are typically spaced uniformly in $\log \tau$).

The resulting deconvolution path can be used to inspect a range of plausible estimates for $f_0$, with varying degrees of smoothness.  This allows the data analyst to bring to bear any further prior information (such as the expected number of modes in $f_0$) that was not formally incorporated into the loss function.  It also enables sensitivity analysis with respect to different plausible assumptions about the smoothness of the mixing distribution.  We illustrate this approach with a real-data example in Section \ref{sec:real_data_example}.

However, in certain cases---for example, in our simulation studies---it is necessary to select a particular value of $\tau$ using a default rule.  We now briefly describe heuristics for doing so based on $\ell_1$  and $\ell_2$ penalties with $k=1$.  These heuristics are used in our simulation studies.  For the case of $\ell_1$ regularization, motivated by \cite{tibshirani2012degrees}, we consider a surrogate AIC approach by computing
\[
\text{AIC}_{\tau} = l(\hat{\theta}_{\tau}) + k  + 1 + \left\vert  \left\{ i :  (\Delta^{(k+1)}\hat{\theta}_{\tau} )_i  \neq 0\right\} \right\vert
\]
and choosing the value of $\tau$ that minimizes this expression. Here, $\hat{\theta}_{\tau}$ denotes the solution given by L1-D with regularization parameter $\tau$.

%\[
%    \underset{\theta}{ \text{minimize  } } \,\,l(\theta) + \tau\,\|D^{(k+1)} \theta\|_1.
% \]

In the case of $\ell_2$ regularization the situation is more difficult, since there is not an intuitive notion of the number of parameters of the model. Instead, we consider an ad-hoc procedure based on cross validation. This solves the problem for a grid of regularization parameters and chooses the parameter the minimizes $l(\hat{\theta}_{\tau}^{\text{held out}}) +  \|\Delta^{(k+1)}\hat{\theta}_{\tau}^{\text{held out}}\|_1$, where  $\hat{\theta}_{\tau}^{\text{held out}}$ is defined as
\[
\hat{\theta}_{\tau}^{\text{held out}} =  \hat{\theta}_{\tau} - \log(n\Delta) + \log(n_{\text{held out}}\Delta)
\]
with $ \hat{\theta}_{\tau}$  the solution obtained by fitting the model on the training set which consists of $75\%$ of the data. Here, $l(\hat{\theta}_{\tau}^{\text{held out}}) $ is evaluated using the counts from the held out set which has $25\%$ of the data, and $n_{\text{held out}}$ is the number of observations in such set. Our motivation for using the additional term $\|\Delta^{(k+1)}\hat{\theta}_{\tau}^{\text{held out}}\|_1$ is that  $\ell_0$ works well when the problem is formulated with $\ell_1$ regularization. However, when (\ref{problem_for_discrete_estimator}) is formulated using $\ell_2$, the penalty $\ell_0$ is not suitable so instead we use $\ell_1$.  Our simulations in the experiments section will show that this rule works well in practice.

\subsection{A toy example}
\begin{figure}[t]
	\includegraphics[width=6in]{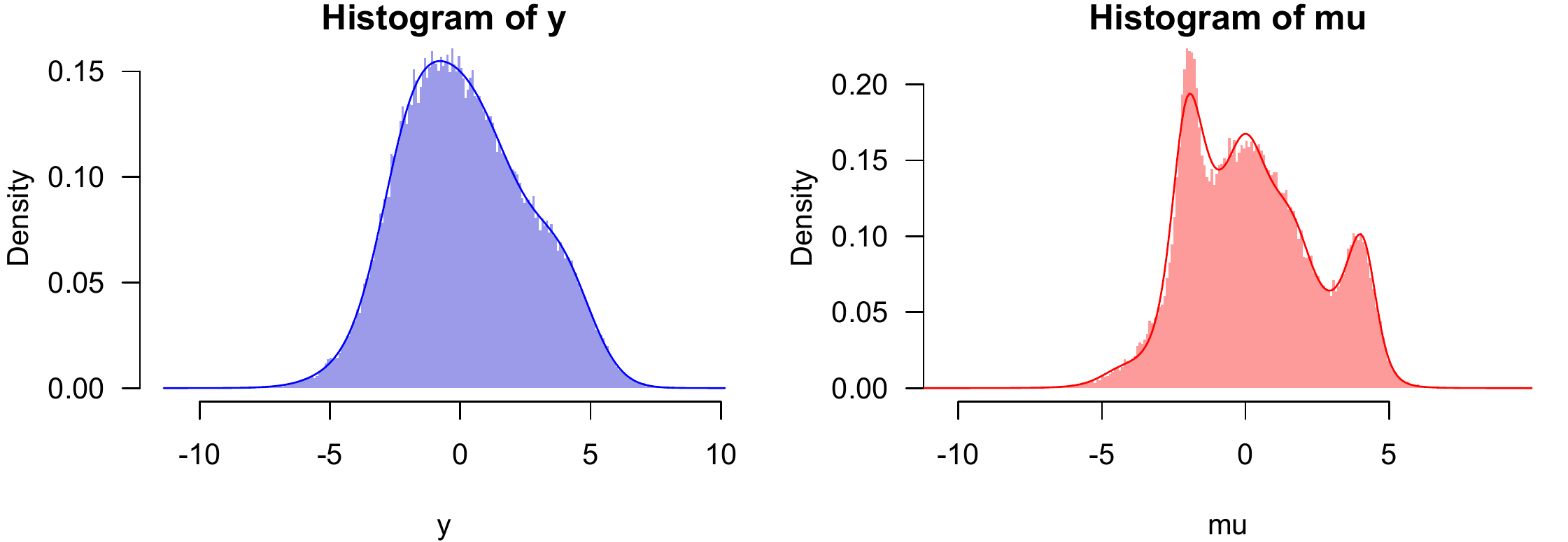}
	\caption{\label{toy_example_1} Example of deconvolution with an $\ell^2$ penalty on the discrete first derivative ($k=1$).  The left panel shows the data histogram together with the fitted marginal density  as a solid curve.  The right panel shows the histogram of the $\mu_i$'s together with the estimated mixing measure as a solid curve.}
\end{figure}

We conclude this section by illustrating the accuracy of our regularized deconvolution approach on a toy example. In this example we draw $10^5$   samples  $\{y_i\}$  with the corresponding $\{\mu_i\}$ drawn from  a mixture of three normal distributions.  Figure \ref{toy_example_1} shows the samples of both the observations $y_i$ (left panel) and the means $\mu_i$ (right panel), together with the reconstructions provided by our method.  Here, we solve the $\ell^2$ version of problem (\ref{problem_for_discrete_estimator}) by using the  BFGS algorithm and choose $\tau$ using the heuristic just described.

It is clear that regularizing with an $\ell_2$ penalty provides an excellent fit of the marginal density.  Surprisingly, it can also capture all three modes of the true mixing density, a feature which is completely obscured in the marginal).  Our experiments in Section \ref{sec:experiments} will show in a more comprehensive way that our method far outperforms other approaches in its ability to provide accurate estimates for multi-modal mixing distributions.

\section{Sensitivity analysis across the path}
\label{sec:real_data_example}

\begin{figure}
	\includegraphics[width=6in]{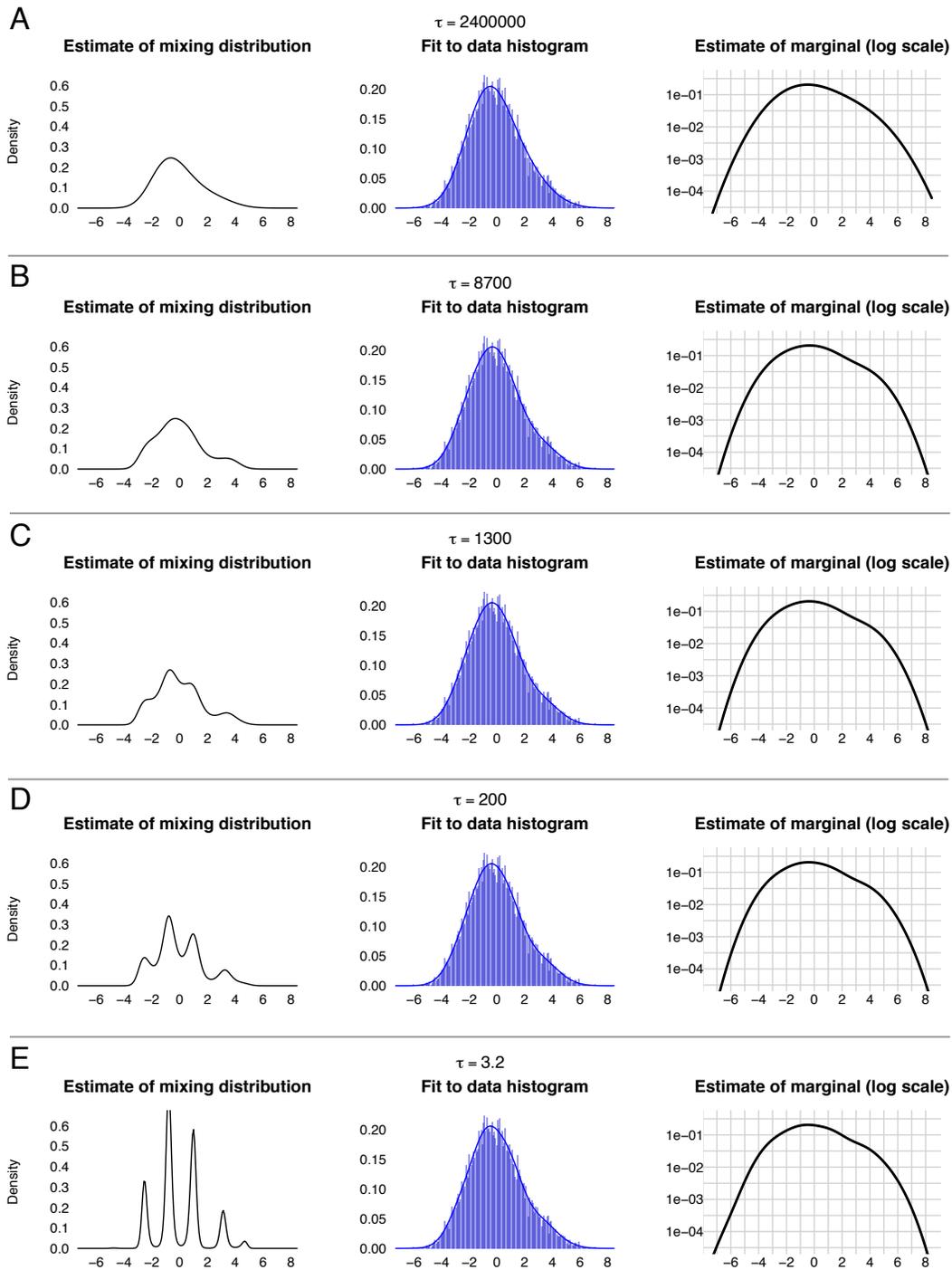}
	\caption{\label{fig:deconvolution_path}Rows A--E show five points along the deconvolution path for the prostate cancer gene-expression data.  The regularization parameter is largest in Row A and gets smaller in each succeeding row.}
\end{figure}

\begin{figure}[t!]
	\includegraphics[width=6in]{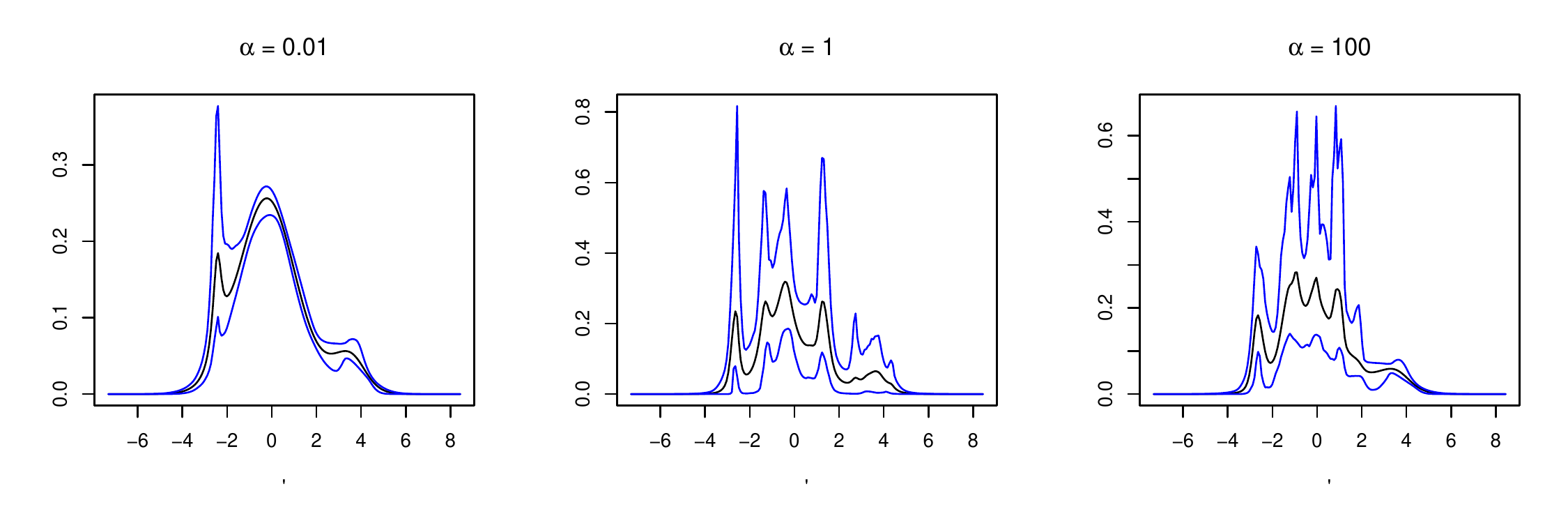}\\
		\includegraphics[width=6in]{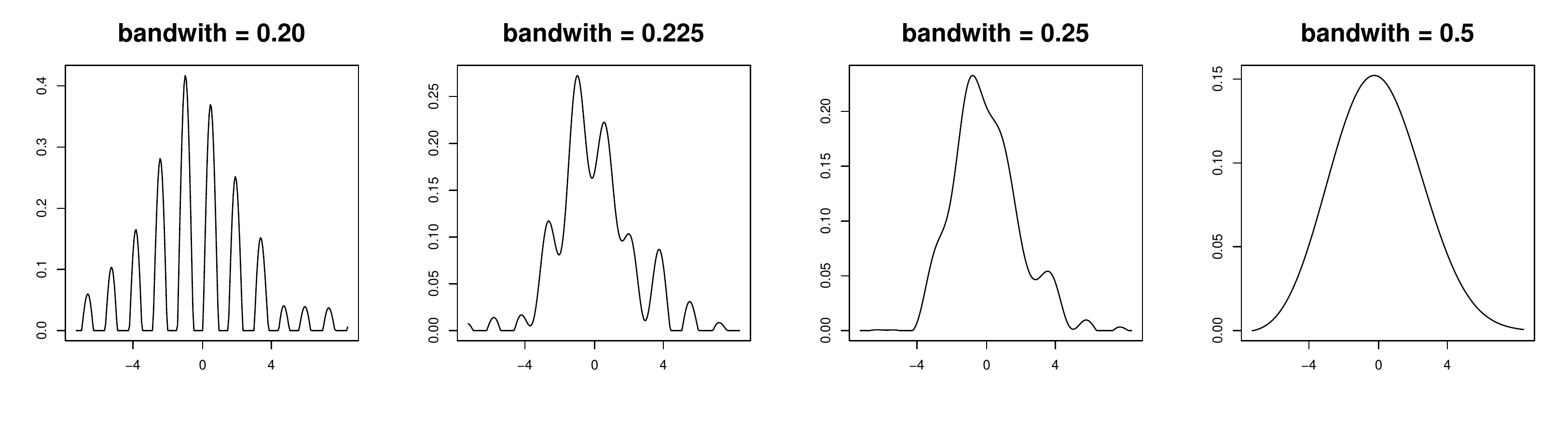}
	\caption{\label{fig:normal_mixture} The first three panels show 95$\%$ confidence bands and posterior mean from  15000 posterior samples from a mixture of 10 normals prior on the latent variables $
		\mu$. Panels 4-7 then shows the estimated mixing density using the kernel estimator with different bandwidth choices. }
\end{figure}

In this section, we provide an example of a sensitivity analysis using our deconvolution path estimator.  We examine data originally collected and analyzed by \citet{singh:etal:2002} on gene expression for 12,600 genes across two samples: 9 healthy patients and 25 patients with prostate tumors.  The data come as a set of 12,600 $t$-statistics computed from gene-by-gene tests for whether the mean gene-expression score differs between the two groups.  After turning these 12,600 $t$-statistics into $z$-scores via a CDF transform, we estimate a deconvolution path assuming a Gaussian convolution kernel.  We use an $\ell^2$ penalty and a grid of $\tau$ values evenly spaced on the logarithmic scale between $10^7$ and $10^{-3}$.

Each row of Figure \ref{fig:deconvolution_path} shows five points along the deconvolution path; the regularization parameter is largest in Row A and gets smaller in each successive row.  Within each row, the left column shows the estimated mixing distribution $\hat f$ for the given value of $\tau$.  The middle column shows the histogram of the data together with the fitted marginal density $\hat m = \phi * \hat f$.  The right column shows the fitted marginal density on the log scale, with a regular grid to facilitate comparison of the results across different values of $\tau$.

The figure shows that, while the estimate of the mixing distribution changes dramatically across the deconvolution path, the estimate of the marginal density is much more stable.  Even on the log scale (right column), the differences among the fitted marginal densities are not visually apparent in Panels B through E, even as the regularization parameter varies across three orders of magnitude.

This vividly demonstrates the well-known fact that deconvolution, especially of a Gaussian kernel, is a very ill-posed inverse problem.  There is little information in the data to distinguish a smooth mixing distribution from a highly multimodal one, and the model-selection heuristics described earlier are imperfect.  A decision to prefer Panel B to Panel E, for instance, is almost entirely due to the effect of the prior.  Yet for most common deconvolution methods, the mapping between prior assumptions and the smoothness of the estimate is far from intuitive.  By providing a full deconvolution path, our method makes this mapping visually explicit.

For reference, it is interesting to compare our deconvolution path to the results of other methods.  Figure \ref{fig:normal_mixture} shows the result of using MCMC to fit a 10-component mixture of normals to the mixing distribution.  The weights in the Gaussian mixture were assigned three different symmetric Dirichlet priors, with concentration parameter $\alpha \in \{0.01, 1, 100\}$.  Panels 1-3 in Figure \ref{fig:normal_mixture} show the posterior mean and posterior 95\% credible envelopes for $f_0$; these settings span a wide range of expected degrees of smoothness for $f_0$, and they yield a correspondingly wide range of posterior estimates.      Comparing figures \ref{fig:deconvolution_path} and \ref{fig:normal_mixture}, we see that the deconvolution path spans essentially the entire range of plausible posterior estimates for $f_0$ arising under any of the concentration parameters. In contrast, Panels  4-7 in Figure show that the kernel estimates are either overly smooth or wiggly.

\section{Theoretical properties}
\label{sec:penalized_likeilhood_approach}

In this section we establish some important theoretical properties of our estimators by thinking of them as approximations to problems involving sieves. We start by showing consistency of the mixing density in L1 norm. We do not provide convergence rates since, unlike the kernel  estimator from \cite{fan1991optimal} and the predictive recursion from \cite{newton2002nonparametric}, our method cannot be expressed in analytical form. Moreover, while the method from \cite{fan1991optimal} remarkably attains minimax rates under squared error loss for point estimates of the true mixing, in an earlier work  \cite{carroll1988optimal} suggested that  convergence  rates for Gaussian deconvolution might be too pessimistic given the unbounded support nature of the classes of functions considered. This is out of the scope of our paper, but we do provide evidence in the later sections that our estimator can outperform existing non-parametric methods. 

%Here we also proof that if the marginal can be well approximated  by an appropriate smoothness sieve, then, under the Hellinger distance, the convergence rate is roughly  $n^{-\frac{k+1}{2(k+2)}}$ where $k$  relates to the smoothness of the first $k+1$ derivatives of the true mixing.

%\subsection{Penalized-likelihood  approach as sieves}

Throughout we consider $k \in \mathbb{N}-\{0\}$ and $q > 0$ to be fixed. We also denote by $\mathcal{P}$ the set of densities in $\mathbb{R}$, thus $\mathcal{P} := \{ f : \int_{\mathbb{R}} f(\mu)d\mu = 1; f\geq  0 \}$, where $d\mu$ denotes Lebesgue measure. Moreover, given any non-negative function $f$ we say that $b \in T_f$ if
\[
\text{max}\left(\| f\|_{\infty}, \| \left(\text{log}f\right)^{(k+1)}\|_{\infty},  \vert \left(\text{log}f\right)^{(k)}(0)\vert, \ldots, \vert \left(\text{log}f\right)(0)\vert  \right )  \leq b,
\]
and $\left(\text{log}f\right)^{(k+1)}$  is  $b$-Lipschitz. Here, given an arbitrary function $g$, we use the notation  $\|\cdot\|_{\infty}$ to indicate  the usual supremum norm on the support of $g$. Moreover $g$  is called    $T_m-$Lipschitz if it satisfies $\vert g(x) \,-\,  g(y)\vert  \leq T_m\,\vert x \,-\,y \vert,$  for all $x$  and $y$.
%\[
%\vert g(x) \,-\,  g(y)\vert  \leq T_m\,\vert x \,-\,y \vert \,\,\forall\,\, x, y.
%\]
%where $T_m$ is an arbitrary positive number.

In this section two metrics of interest will be repeatedly used. The first one is the usual $\ell_1$ distance  $d(f,g)=  \int_{\mathbb{R}} \vert f -g \vert $. The other metric of interest will be the  Hellinger distance whose square is given as $H^2\left(f,g\right) : =   \int_{\mathbb{R}} \vert \sqrt{f}(\mu) - \sqrt{g}(\mu)  \vert^2 d\mu $. We also use the notation $ D_\text{KL}\left(f | g\right) = \int f(\mu)\log (f(\mu)/g(\mu)) d\mu$.
 Finally, for $q \in \mathbb{N}$, we define the functional $J_{k,q}$ which will be a generalization of the usual total variation. We set $J_{k,q}\left(f\right) :=  \int_{\mathbb{R}} \vert f^{(k+1)}\left(\mu\right)\vert^q d\mu $.

Next we state some assumptions for our first consistency result. Our approach is to consider the objective function in   (\ref{eqn:basic_variational_problem}) restricted to a smaller domain than that of its original formulation. This will then allows to prove that the new problem is not ill defined and also its solutions enjoy asymptotic properties of convergence towards the true mixing density. We refer to \cite{geman:hwang:1982} for a general perspective on sieves.
%The idea of studying consistency of estimators for non-parametric maximum likelihood problems over a sequence of reduced spaces is known as the method of sieves and was first introduced in  \cite{geman:hwang:1982}.

\paragraph{Assumptions and definitions}

Let  $A$ be a set of functions that  satisfies the following.

\begin{assumption}\label{as:1}
	Any function $f \in A$ satisfies that  $f \in \mathcal{P}$, $f > 0 $,  $J_{k,q}\left(\log f\right) < \infty$, and there exists a constant $t_f \in T_f$ .
\end{assumption}

\begin{assumption}\label{as:2}
	For all $m \in \mathbb{N}$, the  exists a set $S_m \subset A$  and constants $T_m, K_m > 0$ such that for all $f  \in S_m$  it holds that $	t_f = T_m$   and $ J_{k,q}\left(\log f\right)  \leq K_m$.	Moreover, for all $m$,  the set $S_m $  induces a tight set of of probability measures in $(\mathbb{R},\mathbb{B}(\mathbb{R}))$  satisfying $S_{m} \subset S_{m+1}$. In addition,  $\cup_m \,S_m$ is dense in $A$ with respect  to the metric $d$.
\end{assumption}
%	\[
%t_f = T_m,\,\,   J_{k,q}\left(\log f\right)  \leq K_m, \,\,\,\forall f \in S_m.
%\]

\begin{assumption}\label{as:3}
	Data model:  we assume that  $y_1,\ldots, y_n$  are independent draws from the density $\phi *f_0$,
	$f_0 \in A$, with $\phi$   being an arbitrary density function  satisfying  $\text{max}\left(\|\phi\|_{\infty},\|\phi^{\prime}\|_{\infty}  \right) < \infty$ and $\int_{\mathbb{R}}\log\left(\phi*f_0\right(\mu))\,\phi*f_0(\mu) d\mu < \infty$.  %Moreover, there exists a sequence $f_m $ in $A_m$  such that
\end{assumption}

\begin{assumption}\label{as:4}
	The set
	\[
	A_m = \left\{\alpha \in S_m \,:\,  D_\text{KL}\left(\phi*f_0 \vert\vert  \phi*\alpha\right)  = \underset{\beta \in S_m}{\text{inf }}  D_\text{KL}\left(\phi*f_0 \vert \vert  \phi*\beta\right) \right\},
	\]
	satisfies $d(f_0,\alpha)   \rightarrow 0  \,\,\,\text{as } \,\,\,\,m  \rightarrow \infty$  for all $\alpha \in A_m$, where the convergence is  uniform in $A_m$. 
%	\[
%	\underset{\alpha \in A_m}{\text{sup}} d(f_0,\alpha)   \rightarrow 0  \,\,\,\text{as } \,\,\,\,m  \rightarrow \infty.
%	\]	
\end{assumption}

\begin{assumption}\label{as:5}
	We assume that the  $y_1,\ldots, y_n$  are binned into $D_n$ different intervals with frequency counts   $\{x_j\}_{j=1,\ldots,D_n}$  such that  $n^{-1} \,\|x\|_{\infty}  \rightarrow 0 \,\, a.s.$,
	and we denote by $\xi_j$ an arbitrary point in  interval $j$. Note that this trivially holds for  the case where $D_n = n$ and $x_j = 1$f for all $j$.
\end{assumption}
	%\[
	%\underset{1 \leq j \leq D_n }{\text{max}}\frac{ x_j }{n}  \rightarrow 0 \,\, a.s.
	%\]

\begin{assumption}\label{as:6}
	There exists  $f_m \in A_m$ such that
	\[
	\sum_{j =1 }^{D_n}  \frac{x_j}{n} \text{log}\left(\phi*f_m(\xi_j)\right)  \rightarrow \int_{\mathcal{R}}  \log\left(\phi*f_m(\xi)\right) \phi*f_0(\xi)d\xi   \,\text{  a.s.} \,\,\text{as }\, n \rightarrow \infty.
	\]
	If the $x_j = 1$ and $\xi_j = y_j$  for all $j =1,\dots$, this condition can be disregarded.
\end{assumption}

Assumptions (1)-(3) are natural for the original variational problem proposed earlier. The Lipschitz condition, the bounds on the behavior of the functions at zero, and the tightness of distributions are  merely used to ensure that the sieves will indeed be compact sets with respect to the metric $d$. Moreover, Assumption $(4)$  tell us that the sieves $S_m$ are rich enough to approximate the true mixing density sufficiently well. The last two assumptions can be disregarded when the counts in the bins are all one.
%are technical requirements that

%Before stating our first result, we consider a generic example for which the assumptions above hold.

We are now ready to state our first consistency result. Its proof generalizes ideas from Theorem 1 in \cite{geman:hwang:1982}.
\begin{theorem}
	\label{consistency_1}
	
	If Assumptions (1-6) hold, then, the problem
	$$
	\begin{aligned}
	& \underset{f \in S_m}{\text{minimize}}
	& &
	- \sum_{j=1}^{D_n}  x_j \log \left(\phi*f\right)\left(\xi_j\right),   %\,\,\,\,\,\,\,\,\,\,\,\text{subject to}\,\,\,\,\,\, f \in S_m\\
%	& \text{subject to}
	%& &  f \in S_m
	\end{aligned}
	$$
	has solution set $M^n_m \neq \emptyset$. Moreover, for any sequence $m_n$ increasing slowly enough it holds that
	\[
	\underset{\beta \in M^n_{m_n}}{\text{sup  }} d\left(\beta,f_0\right) \, \rightarrow \, 0 \,\,\text{  a.s}.
	\]
\end{theorem}

In Theorem  \ref{consistency_1}, the sequence $T_{m_n}$ is arbitrary and can grow as fast as desired. Moreover, the a.s statement is on the probability space $\left(\mathbb{R}^{\infty},\mathcal{F},F_0 \times F_0 \times F_0,\ldots\right)$ with $F_0$  the measure  on $\left(\mathbb{R},\mathbb{B}(\mathbb{R})\right) $ induced by $\phi*f_0$, and with $\mathcal{F}$ the completion of $\mathcal{B}\left(\mathbb{R}\right)^{\infty}$.

\section{Experiments}

\subsection{ Mixing density estimation}

\label{sec:experiments}

%new_densities1
\begin{figure}[h!]
	\begin{center}
	\includegraphics[width=5in]{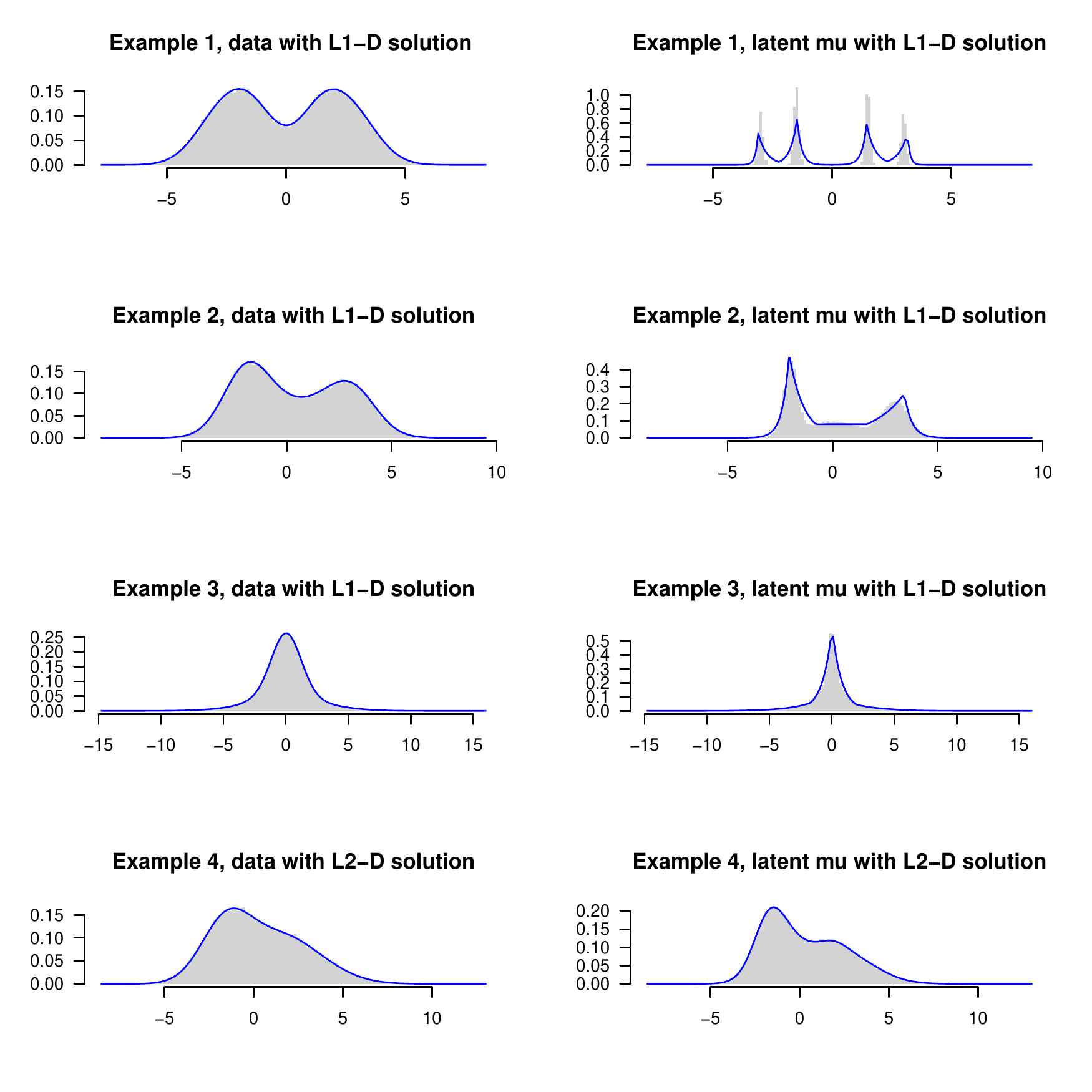}
		\caption{\label{fig:density_exmples} The first panel shows a histogram of observed data $\{y_i\}_{i=1}^n$ for our first example, and the L1-D marginal density estimate plotted on top of the histogram. Here the data has been generated as $y_i \sim  N(\mu_i,1)$ where $\mu_i$  is a draw from the mixing density. The second panel shows, for this same example, the histogram of $\{\mu_i\}_{i=1}^n$ (unobserved draws from the mixing density) and the  L1-D estimate of the mixing density plotted on top of it. Panels 3-6 show the respective cases of Examples 2 and 3. The last two panels show the corresponding plots for the L2-D solution and Example 4.  }  %The first six panels show the L1-D estimated marginal and mixing densities with the observed $y$ and latent $\mu$ given $n = 10^5$ samples for each of examples 1, 2 and 3. The bottom two panels show the corresponding plots for the solutions from L2-D given $n = 10^5$ samples from Example 4.}
	\end{center}
\end{figure}

\begin{figure}[h!]
	\includegraphics[width=6.6in,height =2.6in]{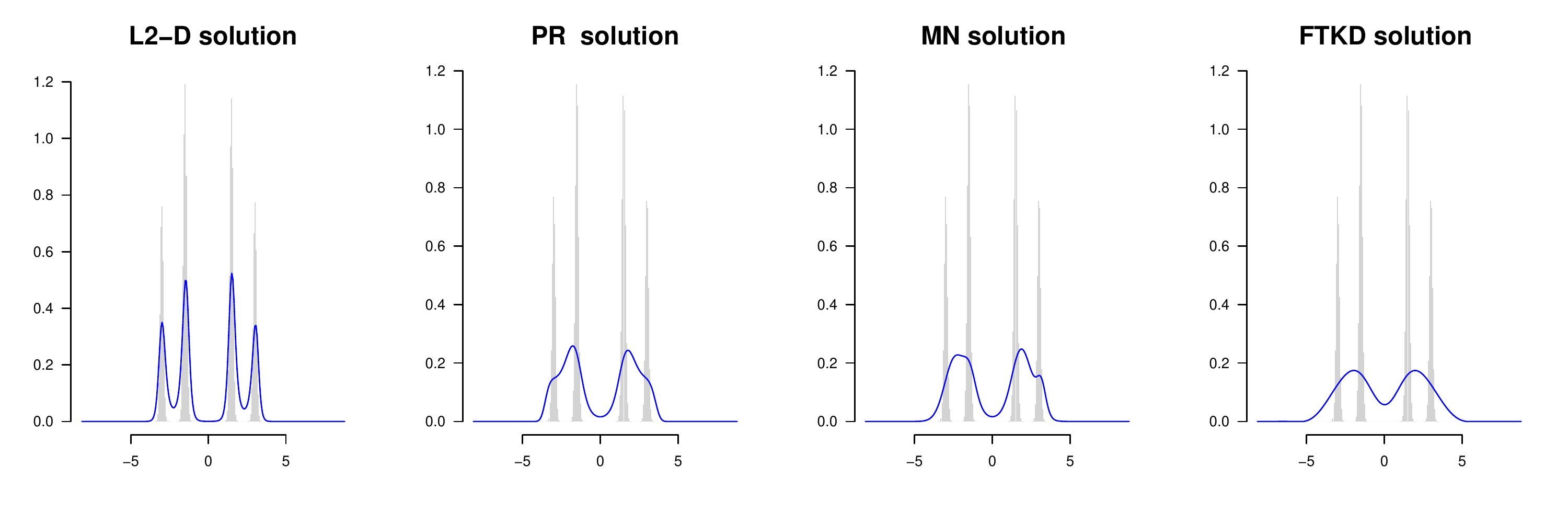}%5.2 before
	\caption{\label{solutions_example1} For the mixing density illustrated in Example 1 of Figure \ref{fig:density_exmples} we show the estimated mixing densities of different methods. The top two panels correspond to the estimated mixing densities using L2-D and PR algorithms along with latent $\mu$. Bottom two panels show the estimated density using MN and FTKD both with the latent $\mu$. For all four panels $n= 10^5$}
\end{figure}

In this section we show the potential gain given by our penalized approaches. We start by considering the task of recovering the true mixing distribution. We evaluate the performance of our methods described in Section \ref{our_approach} which we call L1-deconvolution (L1-D) and L2-deconvolution (L2-D)  depending on the regularization penalty used in the estimation. As competitors we consider   a mixture of normals model (MN), the predictive recursion algorithm (PR) from  \cite{newton2002nonparametric}, and the Fourier transform kernel deconvolution method (FTKD) from \cite{fan1991optimal}. Our comparisons are based on four examples which are shown in Figure \ref{fig:density_exmples}. These examples are  intended to illustrated the performance under different scenarios involving smooth and sharp densities. Next we describe the simulation setting and as well as the implementation details of the competing methods.

As a flexible Bayesian model we decided to use a prior for the mixing density based on a mixture of 10 normals (MN). Here, the weights of the mixture components are drawn from a Dirichlet prior with concentration parameter $1$. This is done in order to have a uniform prior on the simplex. For the locations of the mixture we consider non-informative priors given as $N(0,10^2)$ while for the variances of the mixture components we place a inverse gamma prior with shape parameter $0.01$ and rate $0.01$. The complete model can be then thought as a weak limit  approximation to the Dirichlet process, \citep{ishwaran2002exact}. Also, Gibss sampling   is accomplished straightforwardly by introducing a data augmentation with a variable $z_i$ indicating the component to which $\mu_i$  belong.
%an immediate consequence of conjugacy properties  by introducing a data augmentation with a variable $z_i$ indicating the component to which $\mu_i$  belong.

The next competing model is the predictive recursion algorithm from \cite{newton2002nonparametric} for which we choose the weights $w_i $  as in \cite{martin:tokdar:2012},  close to the limit of the upper bound of the convergence rate for PR given in  \cite{tokdar2009consistency}. Moreover we average the PR estimator over $10$ different permutations of the input data in order to obtain a smaller expected error \cite{tokdar2009consistency}.

On the other hand, for the Fourier transform kernel deconvolution method,  we consider different choices of bandwidth: the rule of thumb from \cite{fan1991optimal}, the plug in bandwidth from \cite{stefanski1990deconvolving}, and the 2-stage plug-in bandwidth from \cite{delaigle2002estimation}. Our estimates are obtained using the R package fDKDE available at \url{http://www.ms.unimelb.edu.au/~aurored/links.html}, which addresses the main concerns associated with the R package decon, see \cite{delaigle2014nonparametric}.

For the final competitor, the ``g-modeling'' approach from \cite{efron2016empirical} (g-M),  we use the newly released  R package deconvolveR.

%report the best solution between the ones provided by choosing the kernel bandwidth  
%we use the R packages decon and  from \cite{wang2011deconvolution} using their default choice of kernel (the support kernel). The bandwith selection is carried with the function bw.dmise which implements a cross validation method inspired by \cite{stefanski1990deconvolving}.

%To summarize, we consider comparisons of our approaches L1-D and L2-D versus a mixtures of 10 normals model (MN), the predictive recursion algorithm (PR) from \cite{newton2002nonparametric}, and the Fourier transform kernel based deconvolution (FTKD) method from \cite{fan1991optimal}.

\begin{table}[t!]
	\centering
	\caption{\label{tab:sim1}Mean squared error (MSE)  between the true and estimated mixing densities, averaging over 100 Monte Carlo simulations, for different methods  given samples from density Example 1. The acronyms here are given the text. The MSE is multiplied by $10^2$ and reported over two intervals containing 95\%  and 99\% of the mass of the mixing density. }
	\medskip
	\begin{small}
		\begin{tabular}{p{3pc} p{2pc} p{2pc} p{2pc} p{2pc}  p{2.5pc}  p{1.7pc}| p{2pc} p{2pc} p{2pc} p{2pc} p{2pc} p{1.5pc}}
			n &  $\begin{array}{l}
				\text{MN} \\
				95\%
			\end{array}$   & $\begin{array}{l}
			\text{PR} \\
			95\%
			\end{array}$     & $\begin{array}{l}
			\text{L2-D} \\
			95\%
			\end{array}$     & $\begin{array}{l}
			\text{L1-D} \\
			95\%
			\end{array}$   & $\begin{array}{l}
			\text{FTKD} \\
			95\%
			\end{array}$ 
			 & $\begin{array}{l}
			 \text{g-M} \\
			 95\%
			 \end{array}$
			  & $\begin{array}{l}
			\text{MN} \\
			99\%
			\end{array}$   & $\begin{array}{l}
			\text{PR} \\
			99\%
			\end{array}$   & $\begin{array}{l}
			\text{L2-D} \\
			99\%
			\end{array}$   & $\begin{array}{l}
			\text{L1-D} \\
			99\%
			\end{array}$  & $\begin{array}{l}
			\text{FTKD} \\
			99\%
			\end{array}$ 
			 & $\begin{array}{l}
			 \text{g-M} \\
			 99\%
			 \end{array}$\\
			\hline
		%	&         & D=250 & D=150 & D=250 & D=150 & D=250 & D=150 & D=250\\
			2000   & 9.47  & 9.12   & 9.39  & 9.69  & \textbf{8.89} &9.27  & 9.26  & 8.91  & 9.18  & 9.48 & \textbf{8.89} &9.01 \\
			10000  & 8.43  & 8.72   & 8.64  & \textbf{7.44}  & 8.87 &9.22 & 8.24  & 8.52  & 8.44  & \textbf{7.28} & 8.87  &9.00\\
			25000  & 8.34  & 8.46   & 7.27  & \textbf{5.54}  & 8.88 &9.32 & 8.15  & 8.27  & 7.13  & \textbf{5.43} & 8.16  &9.09\\
			50000  & 8.21  & 8.23   & 5.80  & \textbf{4.15}  & 8.85 &9.40 & 8.03  & 8.04  & 5.71  & \textbf{4.09} & 8.66  &9.18\\
			100000 & 8.34  & 8.05   & 4.79  & \textbf{3.38}  & 8.69 &9.50  & 8.14  & 7.86  & 4.69  & \textbf{3.35} & 8.49 &9.28\\
		\end{tabular}
	\end{small}
\end{table}

We now state the simulation setting for recovering the mixing density. Given the densities from Figure \ref{fig:density_exmples}, we consider varying the number of samples $n$ and for each fixed $n$ we  run 100 Monte Carlo simulations.   Moreover, for our methods we set $D$, the number of evenly space points in the grid, to 250. See the appendix for a sensitivity example of this parameter.
%points in the grid which are evenly spaced.

The results on Table \ref{tab:sim1}  illustrate a clear advantage of our penalized likelihood approaches over MN, PR  and FTKD which seems even more significant for larger samples size. The estimated mixing density  by L1-D is shown in Figure \ref{fig:density_exmples} where we can clearly see that L1-D can capture the peaks of the unknown mixing density. Moreover, Figure \ref{solutions_example1}  shows that L2-D  can also capture the structure of the true density. In contrast, MN, PR and FTKD all fail to provide reliable estimators.
%present a similar behavior failing to provide reliable estimates. %This is also illustrated in  \ref{solutions_example1} where we see the clear draw back of using PR and MN.

\begin{table}[t!]
	\label{tab:sim2}
	\centering
	\caption{Simulation results for Examples 2, 3 and 4. Mean squared error (MSE)  between the true and estimated mixing densities, averaging over 100 Monte Carlo simulations, for different methods. The acronyms here are given the text. The MSE is multiplied by a constant and reported over two intervals containing 95\%  and 99\% of the mass of the mixing density.}
	\begin{subtable}{1\textwidth}
		\centering
		\caption{\label{tab:sim21} Example 2 results, here each entry of the table is the MSE multiplied by $10^3$. }
		\medskip
		\begin{small}
		\begin{tabular}{p{3pc} p{2pc} p{2pc} p{2pc} p{2pc}  p{2.5pc}  p{1.7pc}| p{2pc} p{2pc} p{2pc} p{2pc} p{2pc} p{1.5pc}}
			n &  $\begin{array}{l}
			\text{MN} \\
			95\%
			\end{array}$   & $\begin{array}{l}
			\text{PR} \\
			95\%
			\end{array}$     & $\begin{array}{l}
			\text{L2-D} \\
			95\%
			\end{array}$     & $\begin{array}{l}
			\text{L1-D} \\
			95\%
			\end{array}$   & $\begin{array}{l}
			\text{FTKD} \\
			95\%
			\end{array}$ 
			& $\begin{array}{l}
			\text{g-M} \\
			95\%
			\end{array}$
			& $\begin{array}{l}
			\text{MN} \\
			99\%
			\end{array}$   & $\begin{array}{l}
			\text{PR} \\
			99\%
			\end{array}$   & $\begin{array}{l}
			\text{L2-D} \\
			99\%
			\end{array}$   & $\begin{array}{l}
			\text{L1-D} \\
			99\%
			\end{array}$  & $\begin{array}{l}
			\text{FTKD} \\
			99\%
			\end{array}$ 
			& $\begin{array}{l}
			\text{g-M} \\
			99\%
			\end{array}$\\
			\hline
				2000  & 6.20   & 2.54   & 2.74  & \textbf{2.38}  & 6.07 &8.23  & 5.47  & 2.37  & 2.49  & \textbf{2.13}  & 5.86 &7.06\\
				10000 & 3.45   & 1.75   & \textbf{1.60}  & 1.68  & 5.98 &5.82 & 3.05  & 1.60  & \textbf{1.46}  & 1.49  &  5.76 &5.45\\
				25000 & 2.31   & 1.46   & \textbf{1.19}  & 1.35  & 5.99 &5.01 & 2.20  & 1.35  & \textbf{1.09}  & 1.19  & 5.70  &4.77\\
				50000 & 1.24   & 1.28   & \textbf{0.89}  & 1.18  & 5.89 &4.68 & 1.10  & 1.17  & \textbf{0.81}  & 1.05  & 5.66  &4.39\\
				100000 & 0.78  & 1.07   & \textbf{0.74}  & 0.87  & 4.97 &4.03 & 0.69  & 0.98  & \textbf{0.67}  & 0.77  & 4.85  &3.85\\
			%	& & & & & & & & & \\
			\end{tabular}
		\end{small}
	\end{subtable}\\
	%\begin{table}[ht]
	\begin{subtable}{1\textwidth}
		\centering
		\caption{\label{tab:sim22} Example 3. Each entry of the table corresponds to the MSE multiplied by $10^3$.}
		\medskip
		\begin{small}
			\begin{tabular}{p{3pc} p{2pc} p{2pc} p{2pc} p{2pc}  p{2.5pc}  p{1.7pc}| p{2pc} p{2pc} p{2pc} p{2pc} p{2pc} p{1.5pc}}
				n &  $\begin{array}{l}
				\text{MN} \\
				95\%
				\end{array}$   & $\begin{array}{l}
				\text{PR} \\
				95\%
				\end{array}$     & $\begin{array}{l}
				\text{L2-D} \\
				95\%
				\end{array}$     & $\begin{array}{l}
				\text{L1-D} \\
				95\%
				\end{array}$   & $\begin{array}{l}
				\text{FTKD} \\
				95\%
				\end{array}$ 
				& $\begin{array}{l}
				\text{g-M} \\
				95\%
				\end{array}$
				& $\begin{array}{l}
				\text{MN} \\
				99\%
				\end{array}$   & $\begin{array}{l}
				\text{PR} \\
				99\%
				\end{array}$   & $\begin{array}{l}
				\text{L2-D} \\
				99\%
				\end{array}$   & $\begin{array}{l}
				\text{L1-D} \\
				99\%
				\end{array}$  & $\begin{array}{l}
				\text{FTKD} \\
				99\%
				\end{array}$ 
				& $\begin{array}{l}
				\text{g-M} \\
				99\%
				\end{array}$\\
				\hline
			%	&       &D=150   & D=250 & D=150 & D=250 & D=150 & D=250 & D=150 & D=250\\
				2000   & 5.28  & 1.83   & 2.09  & \textbf{0.96}  & 4.95 &7.16  & 3.45  & 1.21  & 1.37  & \textbf{0.63} & 3.25&4.54  \\
				10000  & 3.06  & 1.38   & 1.46  & \textbf{0.61}  & 4.86 &1.45 & 1.99  & 0.90  & 0.95  & \textbf{0.40} & 3.18&9.49\\
				25000  & 1.51  & 1.16   & 1.18  & \textbf{0.47}  & 4.61 &1.18 & 0.99  & 0.72  & 0.77  & \textbf{0.31} & 3.01&7.75 \\
				50000  & 0.95  & 1.06   & 1.00  & \textbf{0.42}  & 3.77 &1.13 & 0.62  & 0.70  & 0.66  & \textbf{0.28} & 2.47&9.11\\
				100000 & 0.72  & 0.95   & 0.86  & \textbf{0.38}  & 3.48 &2.42 & 0.47  & 0.62  & 0.56  & \textbf{0.25} & 2.29&1.58\\
			\end{tabular}
		\end{small}
	\end{subtable}\\
	\bigskip
	\begin{subtable}{1\textwidth}
		\centering
		\caption{\label{tab:sim23}Example 4, heach entry of the table is the MSE multiplied by  $10^4$.}
		\medskip
		\begin{small}
		\begin{tabular}{p{3pc} p{2pc} p{2pc} p{2pc} p{2pc}  p{2.5pc}  p{1.7pc}| p{2pc} p{2pc} p{2pc} p{2pc} p{2pc} p{1.5pc}}
			n &  $\begin{array}{l}
			\text{MN} \\
			95\%
			\end{array}$   & $\begin{array}{l}
			\text{PR} \\
			95\%
			\end{array}$     & $\begin{array}{l}
			\text{L2-D} \\
			95\%
			\end{array}$     & $\begin{array}{l}
			\text{L1-D} \\
			95\%
			\end{array}$   & $\begin{array}{l}
			\text{FTKD} \\
			95\%
			\end{array}$ 
			& $\begin{array}{l}
			\text{g-M} \\
			95\%
			\end{array}$
			& $\begin{array}{l}
			\text{MN} \\
			99\%
			\end{array}$   & $\begin{array}{l}
			\text{PR} \\
			99\%
			\end{array}$   & $\begin{array}{l}
			\text{L2-D} \\
			99\%
			\end{array}$   & $\begin{array}{l}
			\text{L1-D} \\
			99\%
			\end{array}$  & $\begin{array}{l}
			\text{FTKD} \\
			99\%
			\end{array}$ 
			& $\begin{array}{l}
			\text{g-M} \\
			99\%
			\end{array}$\\
			\hline
			2000   & 20.6  & 4.75   & \textbf{1.82}  & 3.48  & 3.25 &7.06 & 16.8  & 4.03  & \textbf{1.15}  & 2.88 & 3.00&5.88\\
			10000  & 7.64  & 1.89   & \textbf{0.65}  & 1.93  & 2.87 &4.20 & 6.23  & 1.60  & \textbf{0.53}  & 1.60 & 2.67&3.53\\
			25000  & 2.04  & 1.10   & \textbf{0.48}  & 2.19  & 2.57 &2.34 & 1.67  & 0.93  & \textbf{0.39}  & 1.80 & 2.37&1.95\\
			50000  & 1.03  & 0.69   & \textbf{0.36}  & 1.20  & 2.02 &1.95 & 0.85  & 0.58  & \textbf{0.30}  & 1.00 & 1.86&1.62\\
			100000 & 0.50  & 0.55   & \textbf{0.39}  & 0.90  & 1.36 &1.49 & 0.40  & 0.46  & \textbf{0.32}  & 0.85 & 1.25&1.22\\
			\end{tabular}
		\end{small}
	\end{subtable}
\end{table}

For our example density 2, we observe from Table 2 that in general  L2-D and L1-D offer the best performance. In the case of example 3, we observe that the L1-D again provides better results than the competitors in all the scenarios of sample sizes considered. Even with only 10000 samples L1-D is closer to the true density than all the other methods with more samples. Moreover, L2-D  performs much better than PR and FTKD.  Also, L2-D  seems to be a clear competitor to MN. In the final example density 4, we observe that L2-D is the best method in all the scenarios considered. %generally better results with L2-D  followed by L1-D and MN, whit the other two methods  showing poorer performance.

Overall,  we have shown that for estimating the mixing density, L1-D and L2-D can perform well under different settings, even when other methods exhibit notable deficiencies. The advantage is amplified by the fact that both of our methods  are less computationally intensive that MN,  with L2-D requiring  around 40 seconds to handle problems with $D = 250$, and L1-D under the same problem conditions typically requires  around 5 minutes for a full solution path across 50 values of the tuning parameter.

% with $D= 250$.  for a full solution path of 50 values of the tunning parameter.

% is even more significant since both of our methods are much less computational intensive the MN (or DP), with L2-D requiring a few seconds,  around 8 seconds when $D=150$, to handle problems of  size $n = 10^5$. Whereas L1-D under the same problem conditions  typically requires less than 90 seconds.

\subsection{Normal means estimation}

After evaluating our proposed methodology for the task of estimating the mixing density, we now, for the case of standard normal kernel, focus on the estimation of the normals means $\{\mu_i\}$. For this,  we consider comparisons using the best four among  the methods  used before in addition to other procedures that we  briefly discuss next.

As it is well known (e.g \cite{efron:2011} for description and references ), assuming that the marginal density is known, one can use Tweedie's formula to estimate  $\{\mu_i\}$. For all the methods here this is the approach that we take, except for MN in which case we use the posterior means resulting from Gibss sampling inference. For the methods depending on  grid estimator, the number of bins is set to $250$. %To  make sure that our results are not sensitive to $D$, we also consider experiments with $D=200$ and $D=150$, the results for these values of $D $  can be found in the appendix.

For the method  of  \cite{efron:2011}, we set to 5 the degree of the polynomial approximation to the logarithm of the marginal true density (we found larger values to be less numerically stable). The Poisson surrogate model is then fit in R using the command glm. We also compare against  the  general maximum likelihood empirical-Bayes estimator (GMLEB) from \cite{jiang2009general}, which is a discretized version of the original  Kiefer--Wolfowitz estimator. For our comparisons we use the algorithm proposed in \cite{koenker2014convex}  based on an interior point method algorithm (GMLEBIP). We use the R package REBayes in order to obtain this estimator (\cite{koenker2013rebayes}). On the other hand, for the shape constrained  (SC) estimator from \cite{koenker2014convex}, we rely on a binned count approach based on a weighted likelihood  using R code provided by the authors. Moreover, we consider the estimator from \cite{brown2009nonparametric} using the default choice of bandwidth $h_n = (\log n)^{-1/2}$, which we refer to as BG. The finally competitor is the  non-linear projection (NLP) estimator from \cite{wager2013geometric}.

\begin{table}
	\centering
	\caption{\label{tab:sim3}Mean squared error, of the normal means estimates, times 100 , averaging over 100 Monte Carlo simulations, for different methods  given samples from  example 1. }
	\medskip
	\begin{small}
		\begin{tabular}{p{2.5pc} p{2.5pc} p{2.5pc} p{2.5pc} p{2.5pc} p{2.5pc} p{3.5pc} p{2pc} p{2pc} p{2.5pc}}
			n           & L2-D  & L1-D  & PR    & MN      & Efron  & GMLEBIP  & SC     &  BG       & NLP\\
			2000        & 64.31 & 64.29 & 64.16 & 67.50   & 70.27  & 64.48    & 68.24  & 65.57     & \textbf{64.11}\\			
			10000       & 63.89 & 63.68 & 63.86 & \textbf{63.18}   & 70.00  & 63.80    & 65.56  & 64.06     & 63.27\\
			25000       & 63.52 & \textbf{63.37} & 63.69 & 63.84   & 69.96  & 63.39    & 64.66  & 63.65     &  63.60\\
			50000       & 63.27 & \textbf{63.21} & 63.55 & 65.20   & 69.85  & 63.23    & 64.15  & 63.44     & 63.26\\
			100000      & 63.27 & 63.23 & 63.59 & 63.79   & 69.89  & 63.21    & 63.86  & 63.39     & \textbf{63.18}\\
		\end{tabular}
	\end{small}
\end{table}

From Table \ref{tab:sim3} it is clear that the best methods for example 1 are  L1-D, L2-D, GMLEBIP, and NLP. Moreover, it is not surprising that  GMLEBIP  provides good estimates given that the true mixing  density has mixture components that have small variance. %From the remaining methods that worst performance is obtained for Efron's estimator, presumably because the true mixing density is not smooth. Note also that SC becomes competitive only when $n $  is large but it still worse than GMLEBIP. This agrees with the experiments in \cite{koenker2014convex} where it was found that GMLEBIP generally outperforms SC for the settings in which the mixing distribution is discrete.

For example 2, we can see  from Table \ref{tab:sim4} that again  L2-D and L1-D provide competitive estimates. The other suitable methods for this example seem to be PR and GMLEBIP. With slightly worse estimates MN, BG and SC provide results that are still competitive, with SC being particularly attractive given its computational speed to provide solutions.

\begin{table}
	\centering
	\caption{\label{tab:sim4}Mean squared error, of the normal means estimates, times 100, averaging over 100 Monte Carlo simulations, for different methods  given samples from  example 2.}
	\medskip
	\begin{small}
		\begin{tabular}{p{2.5pc} p{2.5pc} p{2.5pc} p{2.5pc} p{2.5pc} p{2.5pc} p{3.5pc} p{2pc} p{2pc} p{2.5pc}}
			n           & L2-D  & L1-D  & PR    & MN      & Efron  & GMLEBIP  & SC     &  BG       & NLP\\
			2000        & 65.42 & 65.46 & 65.36 & \textbf{64.33}   & 69.97  & 66.20    & 69.60  & 66.99     & 65.75\\			
			10000       & \textbf{64.98} & 65.06 & 65.08 & 65.66   & 69.75  & 65.29    & 67.10  & 65.54     & 65.95\\
			25000       & 65.19 & \textbf{65.08} & 65.32 & 65.21   & 69.94  & 65.12    & 66.42  & 65.49     & 65.09\\
			50000       & \textbf{64.99} & 65.08 & 65.13 & 65.44   & 69.93  & 65.03    & 65.97  & 65.19     & 65.24\\
			100000      & 65.02 & \textbf{64.95} & 65.14 & 65.03   & 69.84  & 65.02    & 65.69  & 65.14     & 64.96\\
		\end{tabular}
	\end{small}
\end{table}

\begin{table}[t!]
	\centering
	\caption{\label{tab:sim5}Mean squared error, of the normal means estimates, times 100, averaging over 100 Monte Carlo simulations, for different methods  given samples from example 3.}
	\medskip
	\begin{small}
		\begin{tabular}{p{2.5pc} p{2.5pc} p{2.5pc} p{2.5pc} p{2.5pc} p{2.5pc} p{3.5pc} p{2pc} p{2pc} p{2.5pc}}
			n           & L2-D  & L1-D  & PR    & MN      & Efron  & GMLEBIP  & SC     &  BG       & NLP\\
			2000        & \textbf{64.99} & 64.96 & 65.41 & 69.05   & 70.54  & 65.74    & 69.70  & 66.99     & 65.77\\			
			10000       & 64.73 & 64.76 & 64.96 & \textbf{64.21}   & 71.34  & 64.85    & 66.92  & 65.36     & 64.81\\
			25000       & \textbf{64.52} & 64.57 & 64.75 & 64.97   & 71.42  & 64.65    & 66.62  & 64.82     & 64.62\\
			50000       & \textbf{64.51} & 64.61 & 64.73 & 65.38   & 71.52  & 64.64    & 66.60  & 64.67     & 64.57\\
			100000      & 64.54 & \textbf{64.41} & 64.76 & 64.54   & 71.96  & 64.56    & 65.17  & 64.62     & 64.46\\
		\end{tabular}
	\end{small}
\end{table}

\begin{table}[t!]
	\centering
	\caption{\label{tab:sim6}Mean squared error, of the normal means estimates, times 100, averaging over 100 Monte Carlo simulations, for different methods  given samples from  example 4.}
	\medskip
	\begin{small}
		\begin{tabular}{p{2.5pc} p{2.5pc} p{2.5pc} p{2.5pc} p{2.5pc} p{2.5pc} p{3.5pc} p{2pc} p{2pc} p{2.5pc}}
			n           & L2-D  & L1-D  & PR    & MN      & Efron  & GMLEBIP  & SC     &  BG       & NLP\\
			2000        & 79.63 & 80.20 & 79.89 & \textbf{78.68}   & 80.00  & 80.97    & 85.47  & 81.58     & 80.01\\			
			10000       & \textbf{79.32} & 79.35 & 79.42 & 79.34   & 79.99  & 79.74    & 82.18  & 79.89     & 79.64\\
			25000       & 79.39 & 79.31 & 79.48 & \textbf{78.79}   & 79.96  & 79.30    & 80.98  & 79.65     & 79.39\\
			50000       & \textbf{79.21} & 79.25 & 79.29 & 79.85   & 79.82  & 79.40    & 80.58  & 79.36     & 79.39\\
			100000      & 79.29 & \textbf{79.22} & 79.37 & 79.51   & 79.91  & 79.30    & 80.15  & 79.37     & 79.36\\
		\end{tabular}
	\end{small}
\end{table}

Finally, for examples 3 and 4 we can see in Tables \ref{tab:sim5} and \ref{tab:sim6} respectively that L1-D and L2-D are the best   or among the best methods in terms of mean squared distance when recovering the unknown means $\mu_i$. Table \ref{tab:sim6} also suggests  that Efron's estimator is more suitable when the true mixing density is very smooth with no sharp peaks. 

%\[ f_1 * f_2(t) \geq \frac16\,t^{1-1/p/-1/q}\, (1+p)^{1/p} (1+q)^{1/q} \|f_1\|_p\|f_2\|_q \]

\section{Discussion}

In many problems in statistics and machine learning, we observe a blurred version of an unknown mixture distribution which we would like to recover via deconvolution.  The main challenge is to find an approach that is computationally fast but still possesses nice statistical
guarantees in the form of rates of convergence. We propose a two-step ``bin-and-smooth'' procedure that achieves both of these goals. This reduces the deconvolution problem to a Poisson-regularized model would can be solved either via standard methods for smooth optimization, or with a fast version of the alternating-direction method of multipliers (ADMM).  Our approach reduces the computational cost compared to a fully Bayesian method and yields a full deconvolution path to illustrate the sensitivity of our solution to the specification of the amount of regularization.  We provide theoretical guarantees for our procedure. In particular, under suitable regularity conditions, we establish the almost-sure convergence of our estimator towards the mixing density.% We also characterize convergence rates for recovery of marginal density and illustrate the type of sensitivity analysis that can be performed in our framework.

There are a number of directions for future inquiry, including multivariate extensions and extensions to multiple hypothesis testing.  These are active areas of current research.

\singlespacing
\begin{small}

	\bibliographystyle{abbrvnat}
	\bibliography{deconvolution_V4}
\end{small}

%\newpage

\appendix

\section{Technical supplement}

\subsection{Gradient expression for $\ell_2$ regularization}

Here we write the mathematical expressions for the gradient of the objective function when performing L2 deconvolution, As in Section 3.3 of the main document. Using the notation there, we have that
$$
[\nabla l(\theta)]_j = \sum_{i=1}^D G_{ij} e^{\theta_j} \left( \frac{x_i}{\lambda_i(\theta)} - 1 \right) \, ,
$$
and
$$
\nabla \Vert \Delta^{(k+1)} \theta \Vert_2^2 = 2 \left(\Delta^{(k+1)} \right)^T \Delta^{(k+1)} \theta \, .
$$

\subsection{Proof of Theorem \ref{consistency_1}  }

\begin{proof}
	
	Motivated by \cite{geman:hwang:1982}, given $\alpha \in A$ we define the function $F\left(\xi,\alpha\right) = \left(\phi*\alpha\right)\left(\xi\right)$ for $\mu \in \mathbb{R}$. Clearly, $F\left(\xi,\alpha\right)$ is a density that induces a measure in $\mathbb{R}$ that is absolutely continuous with respect to the Lebesgue measure in $\mathbb{R}$. Also, we observe that if $\alpha,\beta \in A$, then, for any Borel measurable set $E$,  we have by Tonelli's theorem that
	\[
	\begin{array}{lll}
	\left\vert \int_{E}\phi*\alpha(\mu)d\mu - \int_{E} \phi*\beta(\mu)d\mu     \right\vert %&  = & 
	%\left\vert \int \int_{E} \phi(\mu - y)\alpha(y) d\mu dy  -  \int \int_{E} \phi(\mu   - y)\beta(y) d\mu dy   \right\vert  \\
	&  = &  \left \vert  \int_{\mathcal{R}} \left( \int_{E} \phi(  \mu - y)d\mu \right) \left( \alpha(y) - \beta(y)\right)dy  \right\vert  \\
	& \leq & d(\alpha,\beta).
	\end{array}
	\]
	Hence $d(\alpha,\beta) = 0$  implies that $\phi*\alpha$ and $\phi*\beta$ induce the same probability measures in $(\mathcal{R},\mathcal{B}\left(\mathcal{R}\right))$.
	
	Next we verify the assumptions in Theorem 1 from \cite{geman:hwang:1982}. This is done into different steps below. Steps 1-4  verify the assumptions B1-B4  in Theorem 1 from \cite{geman:hwang:1982}. Steps 5-6 are needed in the general case in which the data is binned. These are  also related to  ideas from \cite{wald1949note}.
	%%geman:hwang:1982Since we are considering a more general objective function, the proof has to be modified.
	
	$\bold{Step }$ 1
	
	Given $\alpha \in A$ and $\epsilon > 0$, the function
	$$
	\xi \rightarrow  \underset{\beta  \in S_m: d\left(\alpha,\beta\right) < \epsilon }{\text{sup}} \left(\phi*\beta\right)(\xi)
	$$
	is continuous and therefore measurable on $\xi$. To see this, simply note that for any $\beta \in  A$  we have that
	\[
	\|\left(\phi*\beta\right)^{\prime}\|_{\infty} \, =  \,  \|\left(\phi^{\prime}*\beta\right)\|_{\infty}  \leq \|\phi^{\prime}\|_{\infty} \int_{\mathcal{R}} \beta(\mu) d\mu \, = \, \|\phi^{\prime}\|_{\infty}.
	\]
	%	The first equality follows by noticing that if $\{x_l\}$  is a sequence converging to $x$, then%
	%	\[
	%	\underset{l \rightarrow \infty}{\text{lim}} \frac{\phi*\beta(x) - \phi*\beta(x_l)}{x -x_l} = \int_{-\infty}^{\infty} \underset{l \rightarrow \infty}{\text{lim}} \frac{ \phi(x - y ) - \phi(x_l -y) }{x - x_l} \beta(y)dy =  \int_{-\infty}^{\infty} \phi^{\prime}(x-y) \beta(y)dy,
	%	\]
	%which holds by the mean value theorem, the  monotone convergence theorem and the fact that $\|\phi^{\prime}\|_{\infty} < \infty$. 
	Hence all the functions $\beta \in S_m$ are $\left(\|\phi^{\prime}\|_{\infty}+1\right)$-Lipschitz and the claim follows. Also, we note that
	\[
	\begin{array}{lll}
	\underset{\epsilon \rightarrow 0}{\text{lim}}  \underset{\beta \in S_m : d\left(\alpha,\beta\right) < \epsilon }{\text{sup}} \left(\phi*\beta\right)(\xi)  =  \phi*\alpha(\xi).\\
	\end{array}
	\]
	This follows by noticing that
	\[
	\begin{array}{lll}
	\left\vert    \underset{\beta \in S_m : d\left(\alpha,\beta\right) < \epsilon }{\text{sup}} \left(\phi*\beta\right)(\xi)  \,-\, \phi*\alpha(\xi)\right\vert  &\leq &     \underset{\beta \in S_m : d\left(\alpha,\beta\right) < \epsilon }{\text{sup}} \left \vert     \left(\phi*(\beta-\alpha)\right)(\xi) \right\vert   \\
	& \leq &  \|\phi\|_{\infty}\underset{\beta : d\left(\alpha,\beta\right) < \epsilon }{\text{sup}} d\left(\alpha,\beta\right) \\
	& \leq & \epsilon\,\|\phi\|_{\infty}.
	\end{array}
	\]
	
	$\bold{Step }$ 2

	Define $E_{\alpha}\left(g\right) := \int_{\mathbb{R}} g(\xi)\left(\phi*\alpha\right)(\xi)d\xi$ for any function $g$. Then for any $\alpha \in A$  and   $\epsilon > 0$  we have
	\[
	\begin{array}{lll}
	E_{f_0}\left( \log \left( \underset{\beta \in S_m: d\left(\alpha,\beta\right) < \epsilon }{\text{sup}} \left(\phi*\beta\right)(\xi) \right)\right)  & \leq & E_{f_0}\left( \log \left( \underset{\beta : d\left(\alpha,\beta\right) < \epsilon }{\text{sup}} \left(\phi*\beta\right)(\xi) \right)\right) \\
	& \leq &  \int_{\mathcal{R}} \text{log}\left(\|\phi\|_{\infty}\right)\,\phi*f_0(\xi)\,d\xi \\
	& < &  \infty.
	
	\end{array}
	\]

	$\bold{Step }$ 3
	
	Next we  show that $S_m$  is compact on $\left(A,d\right)$. Throughout, we use the notation  $\rightarrow_u$ to indicate uniform convergence.  To show the claim, choose $\{\alpha_l\}$ a sequence in $S_m$. Then since $\{ \left(\text{log}(\alpha_l)\right)^{(k+1)}\}$ are $T_m-$Lipschitz
	and uniformly bounded it follows by Arzela-Ascoli Theorem that there exists a sub-sequence $\{\alpha_{1,l}\} \subset \{\alpha_l\}$  such that $\left(\text{log}(\alpha_{1,l})\right)^{(k+1)} \, \rightarrow_u \, g_{k+1}$ in $[-1,1]$ for some function $g_{k+1} : [-1,1] \rightarrow \mathcal{R} $ which is also $T_m-$Lipschitz. Note that we can again use Arzela-Ascoli Theorem applied to the sequence $\{\alpha_{1,l}\}$ to ensure that there exists a sub-sequence $\{\alpha_{2,l}\} \subset \{\alpha_{1,l}\}$ such that  
	$\left( \text{log}(\alpha_{1,l})\right)^{(k+1)} \, $ $\rightarrow_u \,  g_{k+1},$ in $[-2,2]$. Thus we extend the domain of $g_{k+1}$ if necessary.
	
	Proceeding by induction we conclude that for every $N \in  \mathbb{N}$   there exists a sequence $\{\alpha_{N,l}\}_{l \in \mathbb{N}} \subset \{\alpha_{N-1,l}\}_{l \in \mathbb{N}} $  such that
	\[
	\left(\log(\alpha_{N,l})\right)^{(k+1)} \, \rightarrow_u \, g_{k+1},\,\,\,\,\text{as}\,\,\,\,l\,\rightarrow  \,\infty,
	\]
	in $[-N,N]$ as $l \rightarrow \infty$. Hence with Cantor's diagonal argument  we conclude that there exists a sub-sequence $\{\alpha_{l_j}\} \subset \{\alpha_{l}\}$ such that
	\[
	\left(\log(\alpha_{l_j})\right)^{(k+1)} \, \rightarrow_u \, g_{k+1}\,\,\,\,\text{as}\,\,\,\,j\,\rightarrow  \,\infty,
	\]
	in $[-N,N]$ for all $N \in  \mathbb{N}$. Since  $\vert\left(\text{log}(\alpha_{l_j})\right)^{(k)}(0) \vert \leq T_m$ for all $j$. Then without loss of generality, we can assume that
	\[
	\left(\log(\alpha_{l_j})\right)^{(k)} \, \rightarrow_u \, g_{k}\,\,\,\,\text{as}\,\,\,\,j\,\rightarrow  \,\infty,
	\]
	in $[-N,N]$ for all $N \in \mathbb{N}$ and where the function $g_{k}$  satisfies $g_{k}^{\prime} = g_{k+1}$. Continuing with this process we can assume,  without loss of generality, that	
	\[
	\log(\alpha_{l_j}) \, \rightarrow_u \, g\,\,\,\,\text{as}\,\,\,\,j\,\rightarrow  \,\infty,
	\]
	in $[-N,N]$ for all $N \in \mathbb{N}$ for some function $g$  satisfying $g^{(j)} = g_j$ for all $0 \leq j \leq k+1$. Therefore,
	
	\begin{equation}
	\label{step4_e1}
	\alpha_{l_j} \, \rightarrow_u \, \exp(g)\,\,\,\,\text{as}\,\,\,\,j\,\rightarrow  \,\infty,
	\end{equation}
	in $[-N,N]$ for all $N \in  \mathbb{N}$.
	
	Let us now prove that $\exp(g) \in S_m$. First, we observe by the Fatou's lemma  $e^{g}$  is integrable in $\mathcal{R}$  with respect to the Lebesgue measure. since $S_m $ is tight and, we obtain
	\[
	d\left(\exp(g) , \alpha_{l_j}\right) \rightarrow 0.
	\]
	This  clearly also implies that  $\exp(g)$ integrates to $1$ or $\exp(g) \in \mathcal{P}$. Note that also by Fatou's lemma we have that $J_{k,q}(g) \leq K_m$ and  by construction,	
	\[
	\max\left( \|\exp(g)\|_{\infty}, \|g^{(k+1)}\|_{\infty}, \vert g^{(k)}(0) \vert,\ldots,\vert g(0) \vert  \right) \leq   T_m.
	\]
	Finally, combining all of this with $g^{k+1}$ being $T_m$-Lipschitz, we arrive to $\exp(g) \in S_m$.

	$\bold{Step }$  4	
	By assumption (\ref{as:4}), we have that
	\[
	\underset{\alpha \in  A_m}{\text{sup}}  d\left(f_0,\alpha\right)  \rightarrow  0,\,\,\,\,\text{as}\,\,\,\,\,m \rightarrow \infty.
	\]

	$\bold{Step }$  5
	
	Let us show that	
	\[
	\underset{\epsilon \rightarrow 0}{\text{lim  }} E_{f_0}\left(
	\underset{d(\alpha,\beta)<\epsilon,\beta \in S_m }{\text{sup }}\text{log}\left(\phi*\beta\right)
	\right)  =  E_{f_0}\left( \text{log}\left(\phi*\alpha\right) \right)
	\]
	for all $\alpha \in S_m$. First, note that for all $\xi$
	\[
	0 \leq  \text{max}\left\{0,\underset{d(\alpha,\beta)<\epsilon,\beta \in S_m }{\text{sup }}\text{log}\left(\phi*\beta\right)(\xi) \right\}  \leq \text{max}\left\{0,\text{log}\left(\|\phi \|_{\infty} \right)\right\}.
	\]
	Hence, by Step 1 we obtain
	\[
	\begin{array}{l}
	0 \leq \underset{\epsilon \rightarrow 0}{\text{lim  }} E_{f_0}\left( \text{max}\left\{0,\underset{d(\alpha,\beta)<\epsilon,\beta \in S_m }{\text{sup }}\text{log}\left(\phi*\beta\right) \right\} \right)  =  E_{f_0}\left( \text{max}\left\{\text{log}\left(\phi*\alpha \right),0\right\} \right)\\
	\,\,\, \,\,< \infty.     	
	\end{array}
	\]
	Now  we observe that
	\[
	0  \leq  -\text{min}\left\{0,\underset{d(\alpha,\beta)<\epsilon,\beta \in S_m }{\text{sup }}\text{log}\left(\phi*\beta\right)(\xi)  \right\}  \leq -\text{min}\left\{0,\text{log}\left(\phi*\alpha(\xi)\right)\right\},
	\]
	and the claim follows from the monotone convergence theorem.
	
	If $x_j = 1$ and $\xi_j = y_j$  for all $j =1,\dots,D_n$,  the claim of Theorem \ref{consistency_1} follows from Theorem 1 from \cite{geman:hwang:1982}. Otherwise, we continue the proof below. In either case we can see that  the solution set $M_{m}^{n}$ is not empty given that  the map $\alpha \rightarrow \phi*\alpha(\xi)$ is  continuous with respect to the metric $d$ for any $\xi$.
	
	$\bold{Step }$  6
	
	Note that, by Glivenko-Cantelli Theorem  and our assumption on the maximum number of bins, we have that, almost surely, the random distribution	
	\[
	G_n(\xi) =  \sum_{j=1}^{D_n} \frac{x_j}{n}I_{(-\infty,\xi])}(\xi_j)
	\]
	converges weakly to the distribution function associated with $\phi*f_0$. Hence, almost surely, from the Portmanteau theorem we have  for any $\alpha \in S_{m_{n}}$  and any $\delta > 0$ it holds that
	
	\begin{equation}
	\label{e2}
	\begin{array}{lll}
	\underset{ l \rightarrow  \infty }{\text{lim sup }} \mathlarger\sum\limits_{j=1}^{D_l} \frac{x_{j}}{l} \underset{d(\alpha,\beta)<\delta,\beta \in S_m }{\text{sup }}\text{log}\left(\phi*\beta\right)(\xi_j) &  \leq  &  E_{f_0}\left(\underset{d(\alpha,\beta)<\delta,\beta \in S_m }{\text{sup }}\text{log}\left(\phi*\beta\right)(\xi)\right),
	\end{array}
	\end{equation}
	%\underset{d(\alpha,\beta)<\epsilon,\beta \in S_m }{\text{sup }}\text{log}\left(\phi*\alpha\right)(x_j)
	since the function
	$$
	\xi \,\,\rightarrow\,\,\underset{d(\alpha,\beta)<\delta,\beta \in S_m }{\text{sup }}\text{log}\left(\phi*\beta\right)(\xi),
	$$
	is continuous and bounded by above.
	
	Next we define
	\[
	m_1 = \text{min}\left\{m : \underset{\alpha \in A_m }{ \text{sup }}d\left(\alpha, f_0 \right) < \frac{1}{2}\right\}.
	\]
	Clearly,  $\beta_1,\beta_2 \in A_{m_1}$ implies $d(\beta_1,\beta_2) < 1$. Also, we see that the set $\Pi_1 := \{\alpha \in S_{m_1} :  d\left(\alpha,f_0\right) \geq  1 \} \subset S_{m_1} - A_{m_1}$ is d-compact. Hence, there exists $\alpha_1^1, \ldots, \alpha_{h_1}^1 $ in  $\Pi_1$ such that $\Pi_1 \subset \cup_{l = 1}^{h_1} \{\alpha \in \Pi_{1}:  d\left(\alpha,\alpha_l^1\right) < \delta_{1,l} \}$ for positive constants $\{\delta_{1,l}\}$ satisfying that
	\[
	E_{f_0}\left(\underset{d(\alpha,\alpha_l^1)<\delta_{1,l},\alpha \in \Pi_1 }{\text{sup }}\text{log}\left(\phi*\alpha\right) \right)  <   E_{f_0}\left( \text{log}\left(\phi*f_{m_1}\right) \right)
	\]
	for $l = 1,\ldots, h_1$. Therefore from our assumptions on the sets $A_m$  and also from (\ref{e2}), we arrive at
	\[
	\begin{array}{l}
	\underset{ r \rightarrow  \infty }{\text{lim sup }} \mathlarger\sum\limits_{j=1}^{D_r} \frac{x_{j}}{r} \left(\underset{d(\alpha,\alpha_l^1)<\delta_{1,l},\alpha \in \Pi_1 }{\text{sup }}\log\left(\phi*\alpha\right)(\xi_j) -  \log\left(\phi*f_{m_1}(\xi_j)\right) \right)  \\
	\leq     E_{f_0}\left(\underset{d(\alpha,\alpha_l)<\delta_l,\alpha \in \Pi_1 }{\text{sup }}\log\left(\phi*\alpha\right)\right) -E_{f_0}\left( \log\left(\phi*f_{m_1}\right)  \right) \\
	<   0, \,\,\  \text{a.s}.
	\end{array}
	\]
	Hence
	\[
	\underset{ r \rightarrow  \infty }{\text{lim sup }} \mathlarger\sum\limits_{j=1}^{D_r} x_{j} \left(\underset{d(\alpha,\alpha_l^1)<\delta_{1,l},\alpha \in \Pi_1 }{\text{sup }}\log\left(\phi*\alpha\right)(\xi_j) -  \log\left(\phi*f_{m_1}(\xi_j)\right) \right)  =  -\infty \,\,\  \text{a.s}.
	\]
	This implies
	\[
	\underset{ r \rightarrow  \infty }{\text{lim }}  \frac{   \underset{\alpha \in \Pi_1}{\text{sup}}\prod_{j=1}^{D_r}  \phi*\alpha(\xi_j)^{x_j}   }{\prod_{j=1}^{D_r}  \phi*f_{m_1}(\xi_j)^{x_j}}  =  0 \,\,\, \text{a.s.}
	\]
	Next we define
	\[
	k_1 =  \text{min}\left\{k_0 : r \geq k_0 \,\,\text{implies}\,\,    \frac{   \underset{\alpha \in \Pi_1}{\text{sup}}\prod_{j=1}^{D_r}  \phi*\alpha(\xi_j)^{x_j}   }{\prod_{j=1}^{D_r}  \phi*f_{m_1}(\xi_j)^{x_j}}  < 1 \right\}
	\]
	and we set $m_{k_1} = m_1$.
	Therefore,
	\[
	\underset{\alpha \in M_{m_{k_1}}^{k_1} }{\text{sup } } d\left( \alpha, f_{0}\right)  < 1.
	\]
	Next we define $m_2$ as
	\[
	m_2 = \text{min}\left\{m  \geq m_1 : \underset{\alpha \in A_m }{ \text{sup }}d\left(\alpha, f_0 \right) < \frac{1}{4}\right\}  + 1
	\]
	Then,  $\beta_1,\beta_2 \in A_{m_2}$ implies $d(\beta_1,\beta_2) < 1/2$. We also see that $\Pi_2 := \{\alpha \in S_{m_2} :  d\left(\alpha,f_0\right) \geq  1/2 \} \subset S_{m_2} - A_{m_2}$ is d-compact. Hence there exists $\alpha_1^2, \ldots, \alpha_{h_2}^2 $ in  $\Pi_2$ such that $\Pi_2 \subset \cup_{l = 1}^{h_2} \{\alpha :  d\left(\alpha,\alpha_l^2\right) < \delta_{2,l} \}$ for positive constants $\{\delta_{2,l}\}$ satisfying
	\[
	E_{f_0}\left(\underset{d(\alpha,\alpha_l^2)<\delta_{2,l},\alpha \in \Pi_2 }{\text{sup }}\text{log}\left(\phi*\alpha\right)\right)  <   E_{f_0}\left( \text{log}\left(\phi*f_{m_2}\right) \right)
	\]
	for $l = 1,\ldots, h_2$. Therefore,
	\[
	\begin{array}{l}
	\underset{ r \rightarrow  \infty }{\text{lim sup }} \mathlarger\sum\limits_{j=1}^{D_r} \frac{x_j}{r} \left(\underset{d(\alpha,\alpha_l^2)<\delta_{2,l},\alpha \in \Pi_2 }{\text{sup }}\log\left(\phi*\alpha\right)(\xi_j) -  \log\left(\phi*f_{m_2}(\xi_j)\right) \right)  \\
	\leq     E_{f_0}\left(\underset{d(\alpha,\alpha_l^2)<\delta_{2,l},\alpha \in \Pi_2 }{\text{sup }}\text{log}\left(\phi*\alpha\right)\right) -E_{f_0}\left( \text{log}\left(\phi*f_{m_2}\right)  \right) \\
	<   0 \,\,\  \text{a.s}.
	\end{array}
	\]
	So proceeding as before we obtain
	\[
	\underset{ r \rightarrow  \infty }{\text{lim }}  \frac{   \underset{\alpha \in \Pi_2}{\text{sup}}\prod_{j=1}^{D_r}  \phi*\alpha(\xi_j)^{x_j}   }{\prod_{j=1}^{D_r}  \phi*f_{m_2}(\xi_j)^{x_j}}  =  0 \,\,\, \text{a.s}
	\]
	Finally we define
	\[
	k_2 =  \text{min}\left\{k_0  \, : \,  k_0  \geq k_1\,\,\, \text{and} \,\,r \geq k_0 \,\,\text{implies}\,\,    \frac{   \underset{\alpha \in \Pi_2}{\text{sup}}\prod_{j=1}^{D_r}  \phi*\alpha(\xi_j)^{x_j}   }{\prod_{j=1}^{D_r}  \phi*f_{m_2}(\xi_j)^{x_j}}  < 1 \right\}
	\]
	and we set $m_{k} = m_1 $  for all $ k_1 \leq   k < k_2$ and $m_{k_2} = m_2$.  By construction, we have that
	\[
	\underset{\alpha \in M_{m_{k_2}}^{k_2} }{\text{sup } } d\left( \alpha, f_0\right)  < 1/2.
	\]
	Thus an induction argument allow us to conclude the proof.

\end{proof}

\section{Simulation details}

For the cases of the true mixing density we consider four densities of the form
\[
f_0(\mu)  = \sum_{i=1}^{K} w_i\,N(\mu \mid \theta_i, \sigma_i^2 ).
\]
%\end{equation}
In all cases considered here, the observations arise as in (\ref{eqn:normal_means_model})  with a standard normal sampling model. In our first example we evaluate performance for samples of a density that has four peaks or explicitly $K=4$, $w = (0.2,0.3,0.3,0.2)$, $\theta = (-3,-1.5,1.5,3)$ and with small variance $\sigma^2 = (0.01,0.01,0.01,0.01)$. For the second example we consider a mixture of three normals two of which are smooth while the other has a peak. The true parameters in this case are $K=3$, $w = (1/3,1/3,1/3)$, $\theta = (0,-2,3)$ and $\sigma^2 = (2,.1,.4)$. The next example is a mixture of $K = 3$ normals,  one of which has very high variance. The true parameters chosen are  $w = (0.3,0.4,0.3)$, $\theta = (0,0,0)$ and $\sigma^2 = (0.1,1,9)$. Our final example is a mixture,  with $K =3$,  giving raise to a very smooth density, the parameters are $w = (0.5,0.4,.1)$, $\theta = (-1.5,1.5,4)$ and $\sigma^2 = (1,2,2)$. A visualization of these examples is shown in Figure   \ref{fig:density_exmples}.

\newpage

\subsection{Sensitivity to the number of bins}

\begin{figure}[t!]
	\includegraphics[width=6.0in,height=3.2in]{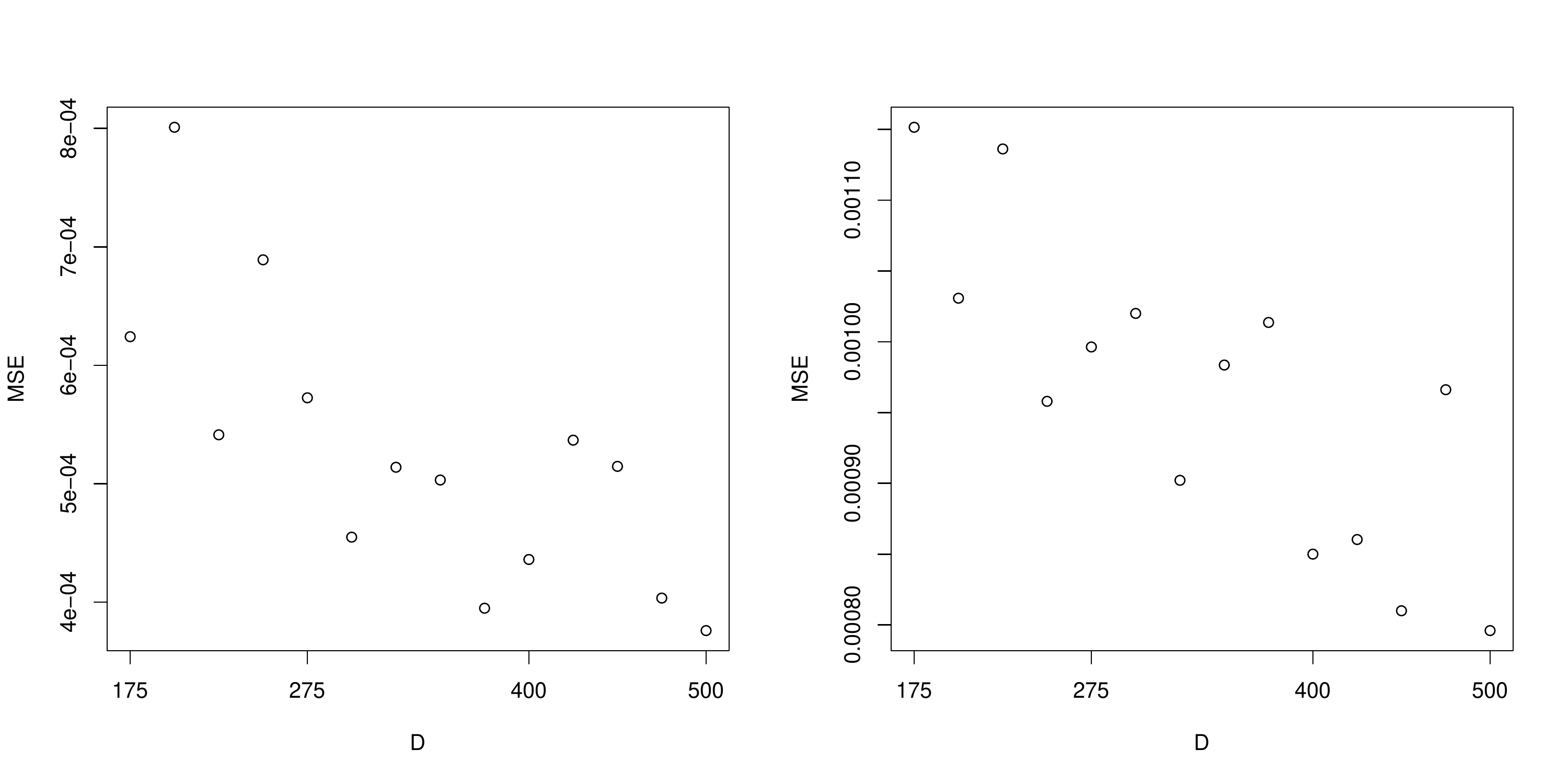}
	\caption{\label{fig:sensitivity_d}Sensitivity of our methods to the parameter $D$. We consider  $n= 10^6$  samples from Example 1  in the paper. The first panel correponds to L2-D  while the second on to L1-D. The results show the average MSE on a $95\%$ mass interval of the true mixing density. The average is taken over 20 MC simulations.}
\end{figure}

Figure \ref{fig:sensitivity_d} shows the performance of both L1-D and L2-D generally improves as we increase $D$. However, based on our experience, $D= 250$ is a reasonable choice. Specially for L1-D whose computational burden increases more rapidly. For $D =250$ it typically takes around $5$ min to compute the solution path for L1-D with 50 values of the regularization parameter. In contrast, L2-D only requires around 40 seconds under the same setting.

\end{spacing}

\end{document}